\documentclass[10pt, journal, letterpaper]{IEEEtran}
\usepackage{bbm}
\usepackage{amssymb}
\usepackage{amsthm}
\usepackage{mathtools}
\usepackage{graphicx}
\usepackage{epstopdf}
\DeclareGraphicsExtensions{.eps}
\usepackage{algorithm, amsmath}
\usepackage{algpseudocode}
\usepackage{float}
\usepackage{subcaption}
\usepackage{etoolbox}
\usepackage[T1]{fontenc}
\usepackage{multirow}
\usepackage{enumerate}

\newcommand{\V}{\mathcal{V}}
\newcommand{\E}{\mathcal{E}}
\newcommand{\N}{\mathcal{N}}
\newcommand{\link}{\mathcal{L}}
\newcommand{\F}{\mathcal{F}}
\newcommand{\U}{\mathcal{U}}

\allowdisplaybreaks

\usepackage{tikz}
\usetikzlibrary{arrows.meta}
\tikzset{>={Latex[width=2mm,length=2mm]}}

\tikzstyle{vertex}=[circle, draw]
\newtheorem{theorem}{Theorem}
\newtheorem{lemma}{Lemma}

\newtheorem{definition}{Definition}
\newtheorem{assumption}{Assumption}

\tikzstyle{stuff_fill}=[vertex, style=green, fill=black!10]
\usetikzlibrary{calc}

\newcounter{defcounter}
\newcounter{mycounter}
\IEEEoverridecommandlockouts

\begin{document}
    
    \title{Joint Placement and Allocation of VNF Nodes with Budget and Capacity Constraints}

\author{Gamal Sallam,~\IEEEmembership{Student Member,~IEEE}~and~Bo Ji,~\IEEEmembership{Senior Member,~IEEE}
\thanks{Gamal Sallam (tug43066@temple.edu) is with the Department of Computer and Information Sciences, Temple University, Philadelphia, PA 19122, USA; Bo Ji (boji@vt.edu) is with the Department of Computer Science, Virginia Tech, Blacksburg, VA 24061, USA. Bo Ji is the corresponding author.
}
\thanks{A preliminary version of this work has been presented at IEEE INFOCOM 2019. This work was supported in part by the NSF under Grant CNS-1651947.}
}
    
    % \author{Gamal Sallam and Bo Ji \\ 
    % Department of Computer and Information Sciences, Temple University, Philadelphia,    PA, USA \\
    % Email: \{tug43066,~boji\}@temple.edu
    %     \thanks{This work was supported in part by the NSF under Grant CNS-1651947. Bo Ji is the corresponding author.}}
    
    \maketitle

    \IEEEpeerreviewmaketitle
    \begin{abstract}
        With the advent of Network Function Virtualization (NFV), network services that traditionally run on proprietary dedicated hardware can now be realized using Virtual Network Functions (VNFs) that are hosted on general-purpose commodity hardware. This new network paradigm offers a great flexibility to Internet service providers (ISPs) for efficiently operating their networks (collecting network statistics, enforcing management policies, etc.). However, introducing NFV requires an investment to deploy VNFs at certain network nodes (called VNF-nodes), which has to account for practical constraints such as the deployment budget and the VNF-node capacity. To that end, it is important to design a joint VNF-nodes placement and capacity allocation algorithm that can maximize the total amount of network flows that are fully processed by the VNF-nodes while respecting such practical constraints. In contrast to most prior work that often neglects either the budget constraint or the capacity constraint, we explicitly consider both of them. We prove that accounting for these constraints introduces several new challenges. Specifically, we prove that the studied problem is not only NP-hard but also non-submodular. To address these challenges, we introduce a novel relaxation method such that the objective function of the relaxed placement subproblem becomes submodular. Leveraging this useful submodular property, we propose two algorithms that achieve an approximation ratio of $\frac{1}{2}(1-1/e)$ and $\frac{1}{3}(1-1/e)$ for the original non-relaxed problem, respectively. Finally, we corroborate the effectiveness of the proposed algorithms through extensive evaluations using trace-driven simulations. 
    \end{abstract}
    \section{Introduction}
    The advent of Network Function Virtualization (NFV) has made it easier for Internet service providers (ISPs) to employ various types of functionalities in their networks. NFV requires the replacement of  network functions that traditionally run on proprietary dedicated hardware with software modules, called Virtual Network Functions (VNFs), which run on general-purpose commodity hardware \cite{chiosi2012network}. A wide variety of network functions (firewalls, intrusion detection systems, WAN optimizers, etc.) can be applied to flows passing through network nodes that host VNFs (called VNF-nodes). A flow must be fully processed at one or multiple VNF-nodes so that the potential benefits introduced by NFV can be harnessed \cite{poularakis2017one}. The new network paradigm enabled by NFV not only offers a great flexibility of introducing new network functions, but it also reduces capital and operational expenditure. Therefore, major ISPs have already started the process of transforming their technologies and operations to support NFV \cite{Amdocs_whitepaper}. 
    
    However, such moves often take place in multiple stages due to the budget limit; in each stage, only a subset of nodes can be selected for deploying/placing VNFs. Moreover, VNF instances typically have a limited capacity, which is shared for processing multiple passing flows. Therefore, given a deployment budget  and capacity limit, it is of critical importance to choose a best subset of nodes to become VNF-nodes  and to determine the optimal capacity allocation so as to maximize the amount of network traffic passing through them. 
    
    In contrast to most prior work that often neglects either the budget constraint (e.g., \cite{ sang2017provably, lukovszki2018approximate}) or the capacity constraint (e.g., \cite{poularakis2017one}), we explicitly consider both constraints and \emph{ formulate a joint problem of VNF-nodes placement and capacity allocation (VPCA)}. The VPCA problem has two main components: VNF-node placement and VNF-node capacity allocation, which are tightly coupled with each other. That is, deciding where to place the VNF-nodes depends on how the capacity of the VNF-nodes will be allocated; determining an optimal capacity allocation apparently depends on where the VNF-nodes are placed. The challenge posed by this problem is two-fold.  First, \emph{the placement and capacity allocation subproblems are both NP-hard}. Second, even if we assume that there is an oracle that can optimally solve the capacity allocation subproblem, \emph{the placement subproblem is non-submodular} (a property that generally leads to efficient solutions for similar problems). This is in stark contrast to the previously studied problem without the capacity constraint \cite{poularakis2017one}, which has been shown to be submodular and can be approximately solved using efficient greedy algorithms.  
    
    To that end, we propose a new \emph{framework} that integrates a \emph{decomposition approach} with a \emph{novel relaxation method}, enabling us to design efficient algorithms with constant approximation ratios for the studied VPCA problem. We summarize our key contributions as follows. 
    %\vspace{-.01in}
    %the above value depends on the situation
    \begin{list}{\labelitemi}{\leftmargin=1em \itemindent=-0.5em \itemsep=.2em}
        \item  First, we formulate the VPCA problem with budget and capacity constraints as an Integer Linear Program (ILP). Then, we provide an in-depth discussion about the new challenges introduced by the budget and capacity constraints. Specifically, we show that the placement and capacity allocation subproblems are both NP-hard. Further, we show that the objective function of the placement subproblem is not submodular.
        \item To address these challenges, we relax the requirement of fully processed flows and allow partially processed flows to be counted. This simple relaxation enables us to prove that the relaxed placement subproblem is submodular based on a \emph{ novel network flow reformulation} of the relaxed capacity allocation subproblem. Leveraging this useful submodular property, we design two efficient algorithms that achieve an approximation ratio of $\frac{1}{2}(1-1/e)$ and $\frac{1}{3}(1-1/e)$ for the original (non-relaxed) VPCA problem, respectively. To the best of our knowledge, \emph{this is the first work that exploits this type of relaxation method to solve a non-submodular optimization problem with provable performance guarantees.}
        
        \item Finally, we evaluate the performance of the proposed algorithms using trace-driven simulations. The simulation results show that the proposed algorithms perform very closely to the optimal solution obtained from an ILP solver and better than another algorithm that iteratively selects the node with the highest volume of traffic traversing it \cite{hong2016incremental}.
    \end{list}
    
    The rest of the paper is organized as follows. First, we position our work compared to related work in Section \ref{sec:related}. Next, we describe the system model and problem formulation in Section \ref{sec:system} and discuss the challenges of the VPCA problem in Section \ref{sec:hardness}. Then, we introduce the VPCA relaxation and reformulation in Section \ref{sec:relaxationAndReformulation} and the proposed algorithms in Section \ref{sec:VPCAALgorithm}. Then, we present the numerical results in Section \ref{sec:evaluation}. Finally, we conclude the paper and discuss future work in Section \ref{sec:conclusion}. 
    \section{Related Work}
    \label{sec:related}
    There has been a large body of work that studies the placement problem in different contexts such as NFV, SDN, and edge cloud computing. In NFV, a placement is usually considered at a scale of VNF instances, i.e., where and how many instances of each network function should be placed and allocated \cite{sang2017provably, shi2018competitive, feng2018optimal, Feng2017}. Different objectives are considered in each of them. The problem of how to meet the demand from all of the flows with a minimum cost (e.g., in terms of the number of instantiated instances) is considered in \cite{sang2017provably,chen2018virtual}. An extension of such work considers the setting where each flow must traverse a chain of network functions, instead of just one function, along a given route~\cite{tomassilli2018provably}. A similar problem is also considered in \cite{shi2018competitive, lukovszki2015online} but for an online setting where flows arrive and leave in an online fashion.  The work in~\cite{Feng2017} addresses the joint problem of VNF service chain placement and routing with the objective of minimizing total communication and computation resource cost. A dynamic version of this problem is considered in~\cite{feng2018optimal}, where the goal is to ensure network stability while minimizing resource cost. Also, in \cite{sallam2018shortest}, the authors consider the placement of a minimum number of nodes to achieve the original maximum flow under a given service function chaining constraint. 
    
    Note that the process of transitioning to NFV typically has two phases: the planning phase and the production phase. The planning phase is concerned about deciding where to introduce NFV to efficiently utilize the limited budget. Since this phase takes place before the actual deployment of the VNF-nodes, one can use historical traces to project flow demands across the network. In the production phase, one can optimize flow admission and routing schemes decisions to efficiently utilize the available resources (e.g., \cite{lukovszki2015online}). In this work, we are mainly focused on the planning phase with budget and capacity constraints and assume that flow routes are fixed. 
    
    There are several studies that are highly relevant to our work. In \cite{poularakis2017one}, the authors consider the selection of a set of nodes to upgrade to SDN. By assuming that the SDN-nodes have an infinite capacity, they show that the problem is submodular. However, we show that with a capacity constraint (which is typically the case in practice), the problem becomes non-submodular. In addition, due to the capacity constraint, only a subset of flows traversing a VNF-node can be processed. Therefore, capacity allocation becomes a crucial component of the joint problem we consider. 
    %Note that although the deployment of SDN-nodes is considered in \cite{poularakis2017one}, the purpose of deploying SDN-nodes and placing VNF-nodes is quite similar in the sense that we enable certain processing to flows that traverse them (e.g., firewall and deep packet inspection). Therefore, the proposed algorithms are of similar flavor.
    Similar to \cite{He2018It}, which considers joint placement and scheduling in the edge clouds, we consider a new architecture with stateless network functions (see, e.g., \cite{kablan2017stateless}), which enables a fractional flow assignment over fixed routes. Similar to our problem, the problem considered in \cite{He2018It} is not submodular in general. While they can prove submodularity and provide an approximation algorithm for a special case, they develop a heuristic algorithm only for the general case.
    In contrast, we develop approximation algorithms for the general problem we consider. Specifically, we propose a new framework that enables us to address the challenge of non-submodularity and develop approximation algorithms for the general case. In \cite{lukovszki2018approximate}, instead of considering a budget constraint, the authors aim to minimize the number of deployed middleboxes subject to the constraint that the length of the shortest path of any flow cannot exceed a certain threshold. They consider homogeneous flow demands (i.e., each flow requires one unit of processing capacity), which makes their capacity allocation subproblem solvable in polynomial time. In contrast, we consider heterogeneous flow demands, which renders the problem NP-hard. Moreover, while the objective function of their placement problem is submodular, ours is non-submodular. Different from the aforementioned studies, we consider both capacity and budget constraints that are of practical importance. Considering these practical constraints introduces new challenges discussed above.
    
    In \cite{lukovszki2018approximate}, the authors also extend their study to the case of heterogeneous flow demands. However, their proposed algorithm achieves a bicriteria approximation ratio only; specifically, the node capacities may be violated by a constant factor. The work of \cite{poularakis2020service} considers the problem of joint service placement and request routing in Mobile Edge Computing networks. This work shows that the considered problem generalizes other well-studied problems, including that of~\cite{lukovszki2018approximate} when flow demands are homogeneous. However, it is unclear whether the generalization applies to the case of heterogeneous flow demands and VNF-node deployment costs, which become relevant in the planing phase addressed in this work. Moreover, the proposed algorithm is based on randomized rounding and achieves a probabilistic bicriteria approximation ratio. Randomized rounding has also been employed to design bicriteria approximation algorithms for Virtual Network Embedding (see, e.g., \cite{rost2019virtual,nemeth2020cost}). Being common in the studies on Virtual Network Embedding, a pre-determined VNF-node placement is often assumed. In contrast to these studies, we consider the problem of joint VNF-node placement and capacity allocation and present algorithms that achieve constant approximation ratios and do not violate any constraint.
    
    The concept of submodularity has been extensively studied in literature, starting with the seminal work in \cite{nemhauser1981maximizing}. Submodular set functions exhibit the diminishing return property, which means that the value of adding an item to a set decreases as the size of the set increases. For problems with submodular objective function, several algorithms can be utilized to solve them efficiently \cite{nemhauser1981maximizing, khuller1999budgeted}. 
    For non-submodular problems, several useful techniques, including weak submodularity \cite{das2011submodular,chen2017weakly} and supermodularity \cite{feldman2014constrained}, have been developed to address non-submodularity. As far as weak submodularity is concerned, a parameter $\gamma \in (0,1]$ is used to quantify how far the objective function is from being submodular. In such settings, approximation results have been established when cardinality constraint \cite{das2011submodular} or general matroid constraint \cite{chen2017weakly} is considered; the guaranteed approximation ratios deteriorate gracefully as $\gamma$ moves away from 1. Since weak submodularity is a relatively new concept, the approximation results remain unexplored when the constraint is of other form, such as knapsack. There is another similar concept called supermodularity \cite{feldman2014constrained}; the supermodular degree is proposed to measure the deviation from submodularity. The work in \cite{feldman2014constrained} introduces an algorithm that can be shown to be effective when the objective function has a small supermodular degree.  Different from these techniques, we propose a novel framework that enables us to design efficient algorithms with constant approximation ratios for the non-submodular problem we consider; the achieved approximation ratios do not depend on problem parameters, such as $\gamma$ in weak submodularity and the supermodular degree. 
    
    Recently, we have also extended our framework to more general settings with multiple network functions and multiple types of resources \cite{sallam2019placement}. There are additional challenges in such general settings: it is unclear whether the relaxed placement subproblem is still submodular; the capacity allocation subproblem becomes a multi-dimensional generalization of the generalized assignment problem with assignment restrictions, which is much more challenging. Due to these new challenges, different algorithms and techniques are developed, and the derived approximation ratios are not constant in general.
    
    \section{System Model and Problem Formulation}
    \label{sec:system}
    We consider a network graph $G=(\V, \E)$, where $\V$ is the set of nodes, with $V = |\V|$, and $\E$ is the set of edges connecting nodes in $G$. We have a set of flows $\F$, with $F = |\F|$. We use $\lambda_f$ to denote the traffic rate of flow $f \in \F$. A node is called a VNF-node if it is able to support VNFs. Since ISPs have a limited budget to deploy VNFs in their networks, they can only choose a subset of nodes $\U \subseteq \V$ to become VNF-nodes. We consider architectures with stateless network functions (e.g., \cite{kablan2017stateless}). A flow's state is stored in a data store; no matter where the flow is processed, the state can be accessed from the data store. Therefore, the traffic rate $\lambda_{f}$ of each flow can be split and processed at multiple VNF-nodes. We use $\lambda_{f}^v$ to denote the portion of flow $f$ that is assigned to VNF-node $v$ and use $\boldsymbol{\lambda} \in \mathcal{R}^{F \times V}$ to denote the assignment matrix. 
    
    We assume that the process of transitioning to NFV goes through two main phases: the planning phase and the production phase. The planning phase is concerned about deciding where to introduce NFV to efficiently utilize the limited budget. Since this phase takes place before the actual deployment of the VNF-nodes, we assume that we can utilize historical traces to project flow demands across the network. Then, in the production phase, we can employ online flow admission and routing schemes (e.g., \cite{lukovszki2015online}) to dynamically adjust flow routing to efficiently utilize the available resources. In this work, we are mainly focused on the planning phase. Therefore, we assume that the traffic of flow $f$ will be sent along a predetermined path, which can be obtained from historical traces. We use $\V_f$ to denote the set of nodes along this path. Alternatively, the nodes along the predetermined path of a flow can also be viewed as potential locations at which the flow will be processed, and routing between these nodes can be dynamically computed in an online fashion. We use $\F_\U$ to denote the set of all flows whose path has one or more nodes in a given set $\U$, i.e., $\F_{\U} = \{f \in \F ~|~ \V_f \cap \U \neq \emptyset\}$.
    
    As we mentioned earlier, the benefits of processed traffic can be harnessed from fully processed flows, i.e., flows that have all of their traffic processed at VNF-nodes. Hence, when a flow traverses VNF-nodes and there is a sufficient capacity on these VNF-nodes to process all of its rate, i.e., $\sum_{v \in \V_f \cap \U} \lambda_{f}^v \geq \lambda_{f}$, then the flow is counted as a processed flow\footnote{For some flows, we can gain benefits by even processing a fraction of its traffic. In such cases, we have a mix of flows that need to be fully processed and that can be partially processed. The key challenges remain, and the proposed algorithms can be applicable with minimal modifications.}. Therefore, the total processed traffic can be expressed as follows:
    \begin{equation}
    \label{eq:objectiveJ_1}
    {J_1}(\U, \boldsymbol{\lambda}) \triangleq \sum_{f\in \F} \lambda_{f} \boldsymbol{1}_{\{ \sum_{v \in \V_f \cap \U} \lambda_{f}^v \geq \lambda_{f} \}},
    \end{equation}
    where $\boldsymbol{1}_{\{.\}}$ is the indicator function. Note that each VNF-node $v$ has a limited processing capacity, denoted by $c_v$. Hence, the total traffic rate assigned to a node should satisfy the following capacity constraint:
    \begin{equation}
    \begin{cases}
    \sum_{f \in \F}  \lambda_{f}^v  \leq c_v,  \, & \forall v \in \U \label{eq:nodecapacity},\\
    \lambda_{f}^v = 0,              & \forall f \in \F \text{ and } \forall v \notin \U.
    \end{cases}
    \end{equation}
    We assume that the largest traffic rate of any flow is no larger than the smallest processing capacity of any node\footnote{While some studies (e.g., \cite{sang2017provably}) consider the placement of VNF instances and allow the flow rate to be larger than the capacity of a VNF instance, we consider the problem of placing VNF-nodes, each of which can host multiple VNF instances. Therefore, it is reasonable to assume that the capacity of such a VNF-node is larger than the rate of any flow.}. Also, we consider a limited budget, denoted by $B$, and require that the total cost of introducing VNF-nodes do not exceed $B$. We use $b_v$ to denote the cost of making node $v$ a VNF-node, which includes hardware and/or software installation cost and may also depend on the processing capacity. Hence, the total cost of VNF-nodes should satisfy the following budget constraint:
    \begin{equation}
    \sum_{v \in \U} b_v \leq B. \label{eq:budget}
    \end{equation}
    
    The above budget constraint limits the number of nodes that can become VNF-nodes, and we may only have a subset of flows that traverse some VNF-nodes. Accounting for the above deployment budget and VNF capacity constraints, we consider a joint problem of VNF-nodes placement and capacity allocation (VPCA). The objective is to choose a best subset of nodes to become VNF-nodes and optimally allocate their capacities so as to maximize the total amount of fully processed traffic. We provide the mathematical formulation of the VPCA problem in the following:
    \begin{equation}\tag{$P1$}
    \label{eq:mainProblem}
    \begin{aligned}
    & \underset{\U \subseteq \V, \boldsymbol{\lambda}}{\text{maximize}} \quad  {J_1}(\U, \boldsymbol{\lambda})\\
    & \text{subject to} \quad \eqref{eq:nodecapacity}, \eqref{eq:budget}. 
    \end{aligned}
    \end{equation}

    \section{Challenges of VPCA}
    \label{sec:hardness}
    Here, we will identify the unique challenges of the VPCA problem formulated in \eqref{eq:mainProblem}. We first decompose the VPCA problem into two subproblems: 1) placement: how to select a subset of nodes to become VNF-nodes and 2) capacity allocation: for a given set of VNF-nodes with fixed capacities, how to divide their capacity for processing a subset of flows. Then, we prove that both subproblems are NP-hard and that the placement subproblem is non-submodular. This is very different from similar problems neglecting the capacity constraint \eqref{eq:nodecapacity} \cite{poularakis2017one}, which have been shown to be submodular and can be approximately solved. 
    
    \subsection{NP-hardness}
    \label{subsec:NPhardness}
    First, we present the formulations of the two subproblems. We start with the allocation subproblem because it will be used in the placement subproblem. For a given set of VNF-nodes $\U \subseteq \V$, let $J^{\U}_2(\boldsymbol{\lambda})$ denote the total amount of fully processed traffic under flow assignment $\boldsymbol{\lambda}$. Note that $J^{\U}_2(\boldsymbol{\lambda})$ has the same expression as that of $J_1(\U, \boldsymbol{\lambda})$ in Eq. \eqref{eq:objectiveJ_1}. The superscript $\U$ of $J^{\U}_2(\boldsymbol{\lambda})$ is to indicate that it is associated with a given set of VNF-nodes $\U$. Then, the capacity allocation subproblem for a given set of VNF-nodes $\U$ can be formulated as
    \begin{equation}\tag{$P2$}
    \label{eq:allocation}
    \begin{aligned}
    & \underset{\boldsymbol{\lambda}: \eqref{eq:nodecapacity}~\text{is satisfied}} {\text{maximize}} \quad  {J_2^\U}( \boldsymbol{\lambda}). 
    \end{aligned}
    \end{equation}
    Let $J_3(\U) \triangleq \max_{\boldsymbol{\lambda}: \eqref{eq:nodecapacity}~\text{is satisfied}} J^{\U}_2 (\boldsymbol{\lambda})$ denote the optimal value of problem \eqref{eq:allocation} for a given set of VNF-nodes $\U$. Then, the placement subproblem can be formulated as
    \begin{equation} \tag{$P3$}
    \label{eq:placement}
    \begin{aligned}
    & \underset{\U \in \V}{\text{maximize}} \quad   {J_3}( \U)\\
    & \text{subject to} \quad \eqref{eq:budget}. 
    \end{aligned}
    \end{equation}
    Note that in order to solve problem \eqref{eq:placement}, we need to solve problem \eqref{eq:allocation} to find the optimal $\boldsymbol{\lambda}$ for a given set of VNF-nodes  $\U$.
    In the following theorem, we will show that both subproblems \eqref{eq:allocation} and \eqref{eq:placement} are NP-hard.
    \begin{theorem}
        \label{theorem:nphardness}
        The capacity allocation subproblem \eqref{eq:allocation} and the placement subproblem \eqref{eq:placement} are both NP-hard.    
    \end{theorem}
    \begin{proof}
    See Appendix~\ref{proof:nphardness}.
    \end{proof}

    \subsection{Non-submodularity}
    \label{subsec:submodularity}

    Note that the objective function $J_3(\U)$ of the placement subproblem \eqref{eq:placement} is a set function. At first glance, problem \eqref{eq:placement} looks like a submodular maximization problem, which has been extensively studied in the literature and can be approximately solved using efficient algorithms \cite{nemhauser1981maximizing, khuller1999budgeted}. However, we will show that the objective function $J_3(\U)$ is generally non-submodular, which makes the placement subproblem \eqref{eq:placement} and the overall problem \eqref{eq:mainProblem} much more challenging. We first give the definition of submodular functions.
    \begin{definition}
        For a finite set of elements $\V$, a function $H: 2^{\V} \rightarrow \mathbb{R}$ is submodular if for any subset $\V_1 \subseteq \V_2 \subseteq \V$ and any element $v \in \V \backslash \V_2$, we have
        \begin{equation}
        \label{eq:submodular}
        H({\V_1 \cup \{v\}}) - H(\V_1)  \geq H({\V_2 \cup \{v\}}) - H(\V_2). 
        \end{equation}
    \end{definition}
     
    The above definition exhibits an important property of diminishing returns. In our problem, if the VNF-node capacity is infinite, i.e., there is no capacity constraint \eqref{eq:nodecapacity}, then a flow $f$ can always be fully processed as long as its path has at least one VNF-node, i.e., $\V_f \cap \U \neq \emptyset$. In this case, the total processed traffic ${J_1}(\U, \boldsymbol{\lambda})$ can be rewritten as
    \begin{equation}
    J^\prime_1(\U) = \sum_{f\in \F} \lambda_{f} \boldsymbol{1}_{\{ \V_f \cap \U \neq \emptyset \}},
    \end{equation}
    where the capacity allocation becomes irrelevant as it does not impact the value of function ${J^\prime_1(\U)}$. It has been shown in \cite{poularakis2017one} that the function ${J^\prime_1}(\U)$  is monotonically nondecreasing and submodular. In this special case, problem \eqref{eq:mainProblem}  with objective function $J^\prime_1(\U)$ can be approximately solved using  efficient greedy algorithms. 
    
        \begin{figure}
        \centering
        \resizebox{0.35\linewidth}{!}{
            \begin{tikzpicture}[transform shape]
            
            \node[vertex](v_1) at (2, 1) {$ v_1 $};
            \node[vertex](v_2) at (4, 0){$ v_2 $};
            \node[vertex ](v_3) at (2, -1){$ v_3 $};

            \node[align=left, color=blue] at (1.6,0) {$f_3$};
            \node[align=left, color=blue] at (3.2, -0.9) {$f_2$};
            \node[align=left, color=blue] at (3.2,.85) {$f_1$};
            \begin{scope}[every path/.style={-}, every node/.style={inner sep=1pt}]

            \path (v_1) edge [ anchor= east] node {} (v_2);
            \path (v_1) edge [ anchor= south] node {} (v_3);
            \path (v_3) edge [ anchor= south] node {} (v_2);        
            \end{scope}
            
            \draw[->] [blue, every node/.style={inner sep=2pt}] plot [smooth, tension=.5] coordinates {  (3.9, -.3)  (2.3, -1.1)}; 
            \draw[->] [blue, every node/.style={inner sep=2pt}] plot [smooth, tension=1] coordinates {  (1.85, -.7) (1.85, .8)};
            \draw[->] [blue, every node/.style={inner sep=2pt}] plot [smooth, tension=1] coordinates { (2.3, 1.1)  (3.8, 0.3) };
            \end{tikzpicture}}
        \caption{An example to show non-submodularity of $J_3(\U)$}
        \label{fig:nonSubmodularity}
    \end{figure}
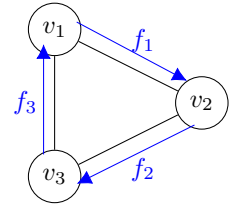
    
    However, using the example presented in Fig. \ref{fig:nonSubmodularity}, we show that the objective function $J_3(\U)$ is no longer submodular if the VNF-nodes have a limited capacity. Consider three flows: flow $f_1$ with path $v_1 \rightarrow v_2$, flow $f_2$ with path $v_2 \rightarrow v_3$, and flow $f_3$ with path $v_3 \rightarrow v_1$. Assume that each VNF-node has a capacity of 3, and each flow has a traffic rate of 2. If node $v_3$ is the only VNF-node, then it can only support one flow because its capacity is 3. Therefore, the marginal contribution of adding node $v_3$ as a VNF-node to the empty set is $J_3(\{v_3\}) - J_3(\emptyset)=2-0=2$. Now, assume that before making node $v_3$ a VNF-node, node $v_2$ is already a VNF-node, which can support one flow. By making node $v_3$ a VNF-node, all three flows can be fully processed, and hence, the total processed traffic becomes 6, i.e., the marginal contribution  of adding node $v_3$ to the set $\{v_2\}$ is $J_3(\{v_2\} \cup \{v_3\}) - J_3(\{v_2\})= 6 - 2= 4 > J_3(\{v_3\}) - J_3(\emptyset) = 2$. This violates the definition of submodular set functions in Eq. \eqref{eq:submodular}. 
    
    As we mentioned in Section \ref{sec:related}, there are several useful techniques that have been developed to handle non-submodular functions, such as weak submodularity \cite{das2011submodular, chen2017weakly} and supermodularity \cite{feldman2014constrained}. 
    When the notion of weak submodularity is considered, one uses a parameter\footnote{For a finite set of elements $\V$, a function $H: 2^{\V} \rightarrow \mathbb{R}$ is $\gamma$-weakly submodular for some $\gamma \in (0, 1]$ if for every two subsets $\V_1, \V_2 \subseteq \V$, we have $\sum_{v \in \V_2} H(\{v\} | \V_1) \geq \gamma H(\V_2 | \V_1)$ \cite{chen2017weakly}.} $\gamma \in (0, 1]$ to measure the deviation of the objective function from being submodular. A larger value of $\gamma$ is better. For the example presented in Fig. \ref{fig:nonSubmodularity}, we can show that the value of $\gamma$ is $0.66$. However, the value of $\gamma$ for function $J_3(\U)$ could be small in general. Even if we can show that $\gamma$ is relatively large, we are still faced with the following issues that hinder the application of weak submodularity to solving the VPCA problem. First, to our best knowledge, the case with knapsack constraint has not been studied yet in literature. Second, it is hard to analyze the value of $\gamma$ in our case. This is due to the fact that the placement value function is the optimal value to the resource allocation subproblem, which is NP-hard (see Theorem~\ref{theorem:nphardness}).
    % even when we have positive results (e.g., when we consider cardinality constraint \cite{das2011submodular}), we can only obtain such results by assuming that the allocation subproblem can be solved optimally, which we have already shown to be NP-hard (Theorem \ref{theorem:nphardness}). We can utilize approximation algorithms for the allocation subproblem, but it is unknown whether combining an approximate resource allocation with weak submodularity can still achieve a constant approximation ratio or not.} 
    On the other hand, the supermodular degree characterizes the level of violation of submodularity for a set function. For problems with a bounded supermodular degree, the authors of \cite{feldman2014constrained} propose a greedy algorithm with performance guarantees for the case with a non-submodular objective function.  However, the proposed greedy algorithm has two main limitations. First, its approximation ratio is  a function of the supermodular degree, which, in our case, could be as large as the number of nodes in the network. Second, its complexity  is exponential in the supermodular degree and could be prohibitively high when the supermodular degree is large. 
    
    Therefore, our problem \eqref{eq:mainProblem} is much more challenging than other similar problems studied in prior work, where the objective function is either submodular or weakly submodular with a large $\gamma$, or there is a bounded supermodular degree. To that end, in the next section we will address the aforementioned unique challenges by introducing a novel relaxation and a problem reformulation, which enable us to propose two algorithms with constant approximation ratios.
    
    \section{Relaxation and Reformulation}
    \label{sec:relaxationAndReformulation}
    In this section, we present  a relaxation of the VPCA problem that allows partially processed flows to be counted. Further, we introduce a novel network flow reformulation of the relaxed capacity allocation subproblem. Both of these techniques will be utilized in designing two efficient approximation algorithms in the next section. 
    \subsection{Relaxed VPCA Formulation}
    \label{subsec:relaxedFormulation}
    We first introduce the relaxed VPCA problem, which allows partially processed flows to be counted. In the relaxed VPCA problem, any fraction of  flow $f$ processed by VNF-nodes in $\V_f \cap \U$ will be counted in the total processed traffic. That is, the relaxed ${J_1}(\U, \boldsymbol{\lambda})$ can be expressed as follows:
    \begin{equation}
    {R_1}(\U, \boldsymbol{\lambda})\triangleq \sum_{f\in \F} \sum_{v\in \V_f \cap \U} \lambda_{f}^v.
    \end{equation}
    Apparently, the total processed traffic of flow $f$ cannot exceed $\lambda_f$, i.e.,  the following constraint needs to be satisfied:
    \begin{equation}
    \sum_{v \in \U} \lambda_{f}^v \leq \lambda_f, \quad \forall f \in \F. \label{eq:traffic2}
    \end{equation}
    Then, the relaxed version of problem \eqref{eq:mainProblem} becomes
    \begin{equation}\tag{$Q1$}
    \label{eq:relaxedProblem}
    \begin{aligned}
    \underset{\U \subseteq \V, \boldsymbol{\lambda}}{\text{maximize}} \quad & {R_1}(\U, \boldsymbol{\lambda})\\
    \text{subject to} \quad & \eqref{eq:nodecapacity}, \eqref{eq:budget}, \eqref{eq:traffic2}. 
    \end{aligned}
    \end{equation}
    
    Next, we decompose problem \eqref{eq:relaxedProblem}, in the same way as we did for problem \eqref{eq:mainProblem}, into placement and allocation subproblems. For a given set of VNF-nodes $\U \subseteq \V$, let $\Lambda^{\U}$ be the set of all flow assignment matrices $\lambda$ that satisfy the capacity constraint \eqref{eq:nodecapacity} and the flow rate constraint \eqref{eq:traffic2}, and let ${R_2^\U}(\boldsymbol{\lambda})$ be the total processed traffic, which has the same expression as that of ${R_1}(\U, \boldsymbol{\lambda})$ but has $\U$ in the superscript so as to indicate that this function is for a given set of VNF-nodes $\U$. Then, the capacity allocation subproblem for a given set of VNF-nodes $\U$ can be formulated as
    \begin{equation}\tag{$Q2$}
    \label{eq:relaxedAllocation}
    \begin{aligned}
    & \underset{\boldsymbol{\lambda} \in \Lambda^{\U}} {\text{maximize}} \quad  {R_2^\U}( \boldsymbol{\lambda}).
    \end{aligned}
    \end{equation}
    
    Now, let ${R_3}(\U) \triangleq \max_{\boldsymbol{\lambda} \in \Lambda^{\U}} {R_2^\U}(\boldsymbol{\lambda})$ denote the optimal value of problem \eqref{eq:relaxedAllocation} for a given set of VNF-nodes $\U$. Then, the placement subproblem can be formulated as
    \begin{equation}\tag{$Q3$}
    \label{eq:relaxedPlacement}
    \begin{aligned}
    & \underset{\U \subseteq \V}{\text{maximize}} \quad   {R_3}( \U)\\
    & \text{subject to} \quad \eqref{eq:budget}. 
    \end{aligned}
    \end{equation}
    
    Note that although the relaxed placement subproblem \eqref{eq:relaxedPlacement} can still be shown to be NP-hard, we will prove that  the objective function ${R_3}(\U)$ is monotonically nondecreasing  and submodular. This useful submodular property allows us to approximately solve problem \eqref{eq:relaxedPlacement}. On the other hand, the relaxed capacity allocation subproblem \eqref{eq:relaxedAllocation} becomes an LP, which can be efficiently solved; alternatively, we can also solve \eqref{eq:relaxedAllocation} using a maximum flow algorithm (discussed at the end of Section~\ref{subsec:placement}). 
    
            \def\x{-2}
    \def\y{0}
    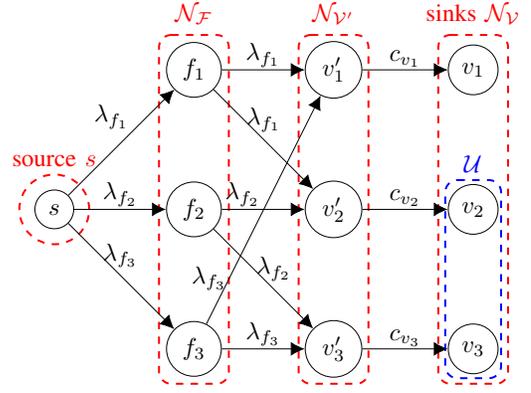
\begin{figure}
        \centering
        \resizebox{0.8\linewidth}{!}{
            \begin{tikzpicture}[transform shape]
            
            %\node[align=left] at (.7,3) {Constructed network graph $Z$};
            \node[vertex](s) at (\x, \y) {$ s $};
            \node[vertex](f1) at (\x+2, \y+2) {$ f_1 $};
            \node[vertex](f2) at (\x+2, \y+0) {$ f_2 $};
            \node[vertex](f3) at (\x+2, \y-2) {$ f_3 $};
            
            \node[vertex](v1prime) at (\x+4, \y+2) {$ v_1^\prime $};
            \node[vertex](v2prime) at (\x+4, \y+0) {$ v_2^\prime $};
            \node[vertex](v3prime) at (\x+4, \y-2) {$ v_3^\prime $};
            
            \node[vertex](v1) at (\x+6, \y+2) {$ v_1 $};
            \node[vertex](v2) at (\x+6, \y+0) {$ v_2 $};
            \node[vertex](v3) at (\x+6, \y-2) {$ v_3 $};

            \draw[red,thick,dashed,rounded corners=5] (\x+2-0.5,\y+2.5) -- (\x+2+0.5,\y+2.5) -- (\x+2+0.5,\y-2.5) -- (\x+2-0.5,\y-2.5) -- cycle;
            \draw[red,thick,dashed,rounded corners=5] (\x+4-0.5,\y+2.5) -- (\x+4+0.5,\y+2.5) -- (\x+4+0.5,\y-2.5) -- (\x+4-0.5,\y-2.5) -- cycle;
            \draw[red,thick,dashed,rounded corners=5] (\x+6-0.5,\y+2.5) -- (\x+6+0.5,\y+2.5) -- (\x+6+0.5,\y-2.5) -- (\x+6-0.5,\y-2.5) -- cycle;
            \draw[blue,thick,dashed,rounded corners=5] (\x+6-0.4,\y-1+1.42) -- (\x+6+0.4,\y-1+1.42) -- (\x+6+0.4,\y-1-1.42) -- (\x+6-0.4,\y-1-1.42) -- cycle;
            %\draw[red,thick,dashed] (\x+2,\y) ellipse  (.72cm and 2.6cm);
            % \draw[red,thick,dashed] (\x+4,\y) ellipse  (.72cm and 2.6cm);
            % \draw[red,thick,dashed] (\x+6,\y) ellipse  (.72cm and 2.6cm);
            % \draw[blue,thick,dashed] (\x+6,\y-1) ellipse  (.6cm and 1.5cm);
            \draw[red,thick,dashed] (\x, \y) circle  (.5cm);
            \node[red, align=left] at (\x,\y+.7) {source $s$};
            \node[red, align=left] at (\x+2,\y+2.8) {$\N_{\F}$};
            \node[red, align=left] at (\x+4,\y+2.8) {$\N_{\V^{\prime}}$};
            \node[red, align=left] at (\x+6,\y+2.8) {sinks $\N_{\V}$};
            \node[blue, align=left] at (\x+6,\y+.65) {$\U$};
            % \node[blue, align=left] at (\x+6.4,\y+.7) {$\U$};
            
            \begin{scope}[every path/.style={->}, every node/.style={inner sep=1pt}]

            \path (s) edge [ anchor= south, pos=0.6, above left] node {$\lambda_{f_1}$} (f1);
            \path (s) edge [ anchor= south, pos=0.5,] node {$\lambda_{f_2}$} (f2);
            \path (s) edge [ anchor= south, pos=0.3, right] node {$\lambda_{f_3}$} (f3);
            
            \path (f1) edge [ anchor= south] node {$\lambda_{f_1}$} (v1prime);
            \path (f1) edge [ anchor= south, pos=0.3,right] node {$\lambda_{f_1}$} (v2prime);
            \path (f2) edge [ anchor= south, near start] node {$\lambda_{f_2}$} (v2prime);
            \path (f2) edge [ anchor= south, pos=0.4, right] node {$\lambda_{f_2}$} (v3prime);
            \path (f3) edge [ anchor= south, pos=0.2, left] node {$\lambda_{f_3}$} (v1prime);
            \path (f3) edge [ anchor= south] node {$\lambda_{f_3}$} (v3prime);
            
            \path (v1prime) edge [ anchor= south] node {${c_{v_1}}$} (v1);
            \path (v2prime) edge [ anchor= south] node {${c_{v_2}}$} (v2);
            \path (v3prime) edge [ anchor= south] node {${c_{v_3}}$} (v3);
            \end{scope}
            
            \end{tikzpicture}}
        \caption{An example of the constructed graph $Z$ for the network in Fig.~\ref{fig:nonSubmodularity}, where $\F = \{f_1, f_2, f_3\}$, $\V = \{v_1, v_2, v_3\}$, $\V_{f_1} = \{v_1, v_2\}$, $\V_{f_2} = \{v_2, v_3\}$, and $\V_{f_3} = \{v_1, v_3\}$}
        \label{fig:netFlow}
    \end{figure}
    
    \subsection{Network Flow Formulation}
    \label{subsec:networkFlowFormulation}
    In this subsection, we introduce a novel network flow reformulation of problem \eqref{eq:relaxedAllocation}. The purpose of this reformulation is two-fold: i) we will use it to prove that the objective function of the relaxed placement subproblem \eqref{eq:relaxedPlacement} is submodular; ii) we will leverage it to develop a combinatorial algorithm for problem \eqref{eq:relaxedAllocation} based on the efficient maxflow algorithms (e.g., \cite{goldberg2014efficient}), which is also a key component of the approximation algorithms we will propose for the original VPCA problem.
    
    For problem \eqref{eq:relaxedAllocation}, we reformulate a network flow problem by constructing a directed graph $Z=(\N, \link)$ as follows. The set of vertices $\N$ consists of the following: an artificial source vertex $s$, set $\N_{\F}$ consisting of flow-vertices $f$ each corresponding to flow $f \in \F$, set $\N_{\V}$ consisting of node-vertices $v$ each corresponding to node $v \in \V$, and set $\N_{\V^{\prime}}$ consisting of node-vertices $v^\prime$ each corresponding to node $v \in \V$. Hence, $\N = \{s\} \cup \N_{\F} \cup \N_{\V} \cup \N_{\V^{\prime}}$, where $\N_{\V}$ consists of the sinks. Let $(x, y)$ be an edge in $\link$, which is from $x \in \N$ to $y \in \N$. The set of edges $\link$ consists of the following: set $\link_1$ consisting of edges $(s, f)$ connecting the source vertex $s$ to each flow-vertex $f \in \N_{\F}$, set $\link_2$ consisting of edges $(f, v^\prime)$ connecting each flow-vertex $f \in \N_{\F}$ to each node-vertex $v^\prime \in \N_{\V^{\prime}}$ corresponding to a node $v \in \V_f$, set $\link_3$ consisting of edges $(v^\prime, v)$ connecting each node-vertex $v^\prime \in \N_{\V^{\prime}}$ to its corresponding node-vertex $v \in \N_{\V}$. We use $c(x, y)$ to denote the capacity of edge $(x, y)$. Hence, $\link = \link_1 \cup \link_2 \cup \link_3$. An edge $(s, f) \in \link_1$ has capacity $\lambda_f$; an edge $(f, v^\prime) \in \link_2$ has capacity $\lambda_f$; an edge $(v^\prime, v) \in \link_3$ has capacity $c_v$. Fig.~\ref{fig:netFlow} presents an example of the constructed graph $Z$ for the network in Fig.~\ref{fig:nonSubmodularity}. 
    
    Next, we describe flows over graph $Z$. Consider functions $\varphi(x, y): \N \times \N \rightarrow \mathbb{R}_+$, where $\mathbb{R}_+$ is the set of non-negative real numbers.
    % Let $\varphi(x, y)$ denote the flow going from vertex $x \in \N$ to vertex $y \in \N$. 
    We define $\Phi(\mathcal{X}, \mathcal{Y}) \triangleq \sum_{x \in \mathcal{X}} \sum_{y \in \mathcal{Y}} \varphi(x, y)$ for $\mathcal{X}, \mathcal{Y} \subseteq \N$. 
    %Let $N_x$ be the set of neighboring vertices of vertex $x \in \V$. 
    An $s$-$\V$ flow is a function $\varphi(x, y): \N \times \N \rightarrow \mathbb{R}_+$ such that the following is satisfied:
    \begin{enumerate}
        \item Capacity constraints: $\varphi(x, y) \leq c(x, y)$ for all pairs $(x, y) \in \N \times \N$. (Note that $c(x,y)=0$ if $(x,y) \notin \link$.)
        \item Flow conservation: the net-flow at every non-source non-sink vertex $x \in \N \setminus (\{s\} \cup \N_{\V})$ is zero, i.e., $\Phi(\N, \{x\}) - \Phi(\{x\}, \N)$ = 0.
        \item Positive incoming flow: the net-flow at the source $s$ is non-positive, i.e., $\Phi(\N, \{s\}) - \Phi(\{s\}, \N) \leq 0$.
        \item Positive outgoing flow: the net-flow at every sink $t \in \N_{\V}$ is non-negative, i.e., $\Phi(\N, \{t\}) - \Phi(\{t\}, \N) \geq 0$.
    \end{enumerate}
    Let $\overline{\F}$ be the set of all $s$-$\V$ flows over $Z$.

    For a subset of sinks\footnote{Note that each node $v \in \V$ corresponds to a sink in $\N_{\V}$. Hence, by slightly abusing the notations, for any $\U \subseteq \V$, we also use $\U$ to denote the corresponding subset of sinks in $\N_{\V}$.} $\U \subseteq \N_{\V}$, we define
    \begin{equation}
    \label{eq:net_flow}
        F(\U) \triangleq \max_{\varphi \in \overline{\F}} (\Phi(\N, \U) - \Phi(\U, \N)),
    \end{equation}
    which is the maximum total net-flow at the sinks in $\U$. The maximum net-flow problem is to find an $s$-$\V$ flow (i.e., function $\varphi$) that achieves the maximum in \eqref{eq:net_flow}.
     In Lemma~\ref{lemma:Q2MaxFlow}, we show the equivalence between the capacity allocation subproblem \eqref{eq:relaxedAllocation} and the maximum net-flow problem \eqref{eq:net_flow}.

    \begin{lemma}
        \label{lemma:Q2MaxFlow}
        The capacity allocation subproblem \eqref{eq:relaxedAllocation} is equivalent to the maximum net-flow problem \eqref{eq:net_flow}. 
        Hence, for any given $\U \subseteq \V$, the optimal value of problem \eqref{eq:relaxedAllocation} is equal to the maximum total net-flow at the sinks in $\U \subseteq \N_{\V}$ of the associated graph $Z$, i.e., 
         \begin{equation}
         \label{eq:equivalence}
         R_3(\U) = F(\U).
         \end{equation}
    \end{lemma}
    \begin{proof}
    See Appendix~\ref{proof:Q2MaxFlow}.
    \end{proof}

    \section{Proposed Algorithms}
    \label{sec:VPCAALgorithm}
    In this section, we design two efficient algorithms that can achieve constant approximation ratios for the VPCA problem \eqref{eq:mainProblem}. The main idea is to utilize the relaxation introduced in the previous section, which allows partially processed flows to be counted. By doing so, we can show that the relaxed placement subproblem is submodular based on the network flow reformulation of the relaxed capacity allocation subproblem. In this case, the relaxed placement subproblem can be approximately solved using efficient greedy algorithms. Moreover, the relaxed allocation subproblem becomes a Linear Program (LP), which can also be solved efficiently in polynomial time. However, the solution to the relaxed problem is for the case where any fraction of the processed flows is counted. In order to obtain a solution for the original VPCA problem \eqref{eq:mainProblem}, where only the fully processed flows are counted, we propose two approximation algorithms by modifying the solution to the relaxed capacity allocation subproblem: the first one is based on a maximum flow algorithm, and the second one is based on a greedy algorithm. 
    
    We use RP-MCA and RP-GCA to denote the algorithms we develop by combining the Relaxed Placement with the Maximum-flow-based Capacity Allocation and the Greedy Capacity Allocation, respectively. We show that the RP-MCA and RP-GCA algorithms achieve an approximation ratio of $\frac{1}{2}(1-1/e)$  and $\frac{1}{3}(1-1/e)$, respectively. We describe the algorithms in a unified framework presented in Algorithm \ref{alg:proposed}. The difference is in the capacity allocation subproblem (line \ref{line:capacityAllocation}), where RP-MCA algorithm uses a Max-flow-based Capacity Allocation (MCA) algorithm presented in Algorithm \ref{alg:MCA}, while RP-GCA algorithm uses a Greedy Capacity Allocation (GCA) algorithm presented in Algorithm \ref{alg:GCA}.
    
    \begin{algorithm}[t]
        \caption{The RP-MCA  and RP-GCA algorithms}
        \label{alg:proposed}
        \begin{algorithmic}[1]
            \Statex {Input: set of nodes $\V$, set of flows $\F$, node capacities, node costs, flow rates, and budget $B$.}
            \Statex Output: set of VNF-nodes $\U$, capacity allocation $\boldsymbol{\lambda}$.
            \State \textbf{Relaxed Problem}: relax function ${J_1}(\U, \boldsymbol{\lambda})$ to become ${R_1}(\U, \boldsymbol{\lambda})$;
            \State \textbf{Placement Subproblem}: solve problem \eqref{eq:relaxedPlacement} using the Submodular Greedy algorithm  or the Enumeration-based Greedy algorithm, described in Section \ref{subsec:placement}, to obtain $\U$.
            \State \textbf{Capacity Allocation}: use either the MCA algorithm (Algorithm \ref{alg:MCA}) or the GCA algorithm (Algorithm \ref{alg:GCA}) to obtain capacity allocation $\boldsymbol{\lambda}$. \label{line:capacityAllocation}
        \end{algorithmic}
    \end{algorithm}
   
    \subsection{Proposed Placement Algorithms}
    \label{subsec:placement}
    In this subsection, we first prove in Lemma \ref{lemma:submodularity} that the objective function $R_3(\U)$ of the relaxed placement subproblem \eqref{eq:relaxedPlacement} is monotonically nondecreasing and submodular. Then, using the property of submodularity, we propose two greedy algorithms for solving the placement subproblem. 
    \begin{lemma}
        \label{lemma:submodularity}
        The function ${R_3}(\U)$ is monotonically nondecreasing and submodular.
    \end{lemma}
            \begin{proof}
        The function $R_3(\U)$ is monotonically nondecreasing because adding an additional VNF-node does not reduce the amount of flows that can be fully processed. 
        
        Next, we prove that the function ${R_3}(\U)$ is submodular. The proof follows from the network flow reformulation introduced in Section~\ref{subsec:networkFlowFormulation}. Applying Lemma~\ref{lemma:Q2MaxFlow} and the max-flow min-cut theorem (see, e.g., \cite[pp.~348--349]{bach2013learning}) immediately gives
        \begin{equation}
        {R_3}(\U) = \min_{\mathcal{X} \subseteq \N: s \in \mathcal{X}, \, \U \subseteq \N \setminus \mathcal{X}} \sum_{x \in \mathcal{X}} \sum_{y \in \N \setminus \mathcal{X}} c(x, y).
        \end{equation}
        One then obtains that $R_3(\U)$ is submodular (as the partial minimization of a cut function) \cite[p.~230]{bach2013learning}. 
    \end{proof}
    
    Because of this useful submodular property, problem \eqref{eq:relaxedPlacement} can be approximately solved using  efficient greedy algorithms. Next, we consider two cases of problem \eqref{eq:relaxedPlacement}: uniform VNF-node costs (Case I, a special case) and heterogeneous VNF-node costs (Case II, a general case). 
    
    In Case I, the VNF-nodes have uniform costs, i.e.,  $b_v = b$ for all $v \in \V$. Then, the budget constraint \eqref{eq:budget} can be expressed as a cardinality constraint, i.e., $|\U| \leq k$, where $k =\lfloor B/b \rfloor$. In this case, we can use a simple \emph{Submodular Greedy} (SG)  algorithm  to approximately solve problem \eqref{eq:relaxedPlacement}. In the SG algorithm, we start with an empty solution of VNF-nodes $\U$; in each iteration, we add a node that has the maximum marginal contribution to $\U$, i.e., a node that leads to the largest increase in the value of the objective function. If multiple nodes have the same marginal contribution, we break ties by selecting a node uniformly at random. We repeat the above procedure until $k$ VNF-nodes have been selected. This solution has been shown to achieve  an approximation ratio of  $(1-1/e)$ \cite{nemhauser1981maximizing}. However, this algorithm does not guarantee to have the same approximation ratio for the case of heterogeneous VNF-node costs \cite{khuller1999budgeted}. 
    
    In Case II, the VNF-nodes have heterogeneous costs, i.e., the costs of VNF-nodes are different. For this case, an \emph{Enumeration-based Greedy} (EG) algorithm has been proposed in \cite{khuller1999budgeted}, which can be shown to achieve the same approximation ratio of  $(1-1/e)$, but with a higher running time complexity compared to the SG algorithm. The EG algorithm has two phases. In Phase I, it samples all node subsets of cardinality one or two that satisfy the budget constraint, picks the one with the largest value of the objective function ${R_3}$, and stores this temporary solution in $\U_1$. In Phase II, the algorithm samples all node subsets of cardinality three and augments each of these subsets with nodes that maximize the relative marginal contribution $({R_3}(\V^\prime \cup \{u\}) - {R_3}(\V^\prime))/b_u$, in a greedy manner. The budget constraint must also be satisfied throughout this procedure. Then, it selects the  augmented subset with the largest value of the objective function ${R_3}$ and stores it in $\U_2$. The final solution will be the better one between $\U_1$ and $\U_2$, i.e., the one that achieves a larger value of the objective function ${R_3}$.
    
    Note that although the value of function $R_3(\U)$ can be obtained using an LP solver, we can alternatively compute it using the network flow formulation presented in Section \ref{subsec:networkFlowFormulation} as follows. For the constructed graph $Z$, we connect all the sink vertices corresponding to nodes $\U$ to an artificial sink vertex $d$. Then, the value of $R_3(\U)$ is the maximum flow from vertex $s$ to vertex $d$ in graph $Z$, which can be computed using several efficient algorithms (see, e.g., \cite{goldberg2014efficient}). In Lemma \ref{lemma:placement}, we restate the results of \cite{nemhauser1981maximizing, khuller1999budgeted} about the approximation ratio of the SG and EG algorithms.
    \begin{lemma}
        \label{lemma:placement}
        Both the SG and EG algorithms achieve an approximation ratio of $(1-1/e)$.
    \end{lemma}
    \begin{proof}
        The proofs can be found in \cite{nemhauser1981maximizing} and \cite{khuller1999budgeted} for the SG algorithm and the EG algorithm, respectively.
    \end{proof}
    
    \subsection{ Proposed Capacity Allocation Algorithms}
    \label{subsec:capacityAllocation}
    While the solution of problem \eqref{eq:relaxedPlacement} allows partially processed flows to be counted, only fully processed flows will be counted in the original problem \eqref{eq:mainProblem}. To that end,  we propose two algorithms to modify the capacity allocation of VNF-nodes $\U$ so as to ensure fully processed flows and provide certain performance guarantees. The first algorithm is based on the network flow formulation, and the second one is based on a simple greedy approach. We develop these algorithms by modifying two algorithms for the multiple knapsack problem with assignment restrictions (MKAR) \cite{dawande2000approximation}. However, we want to point out that there is a key difference between our studied VPCA problem and the MKAR problem: in the VPCA problem, a flow can be split and assigned to more than one VNF-node, while in the MKAR problem, an item (corresponding to a flow in our problem) cannot be split and must be assigned to at most one knapsack (corresponding to VNF-node in our problem). Because of this key difference, an optimal solution for the VPCA problem generally has a larger value compared to that of the MKAR problem. Therefore, the algorithms developed for the MKAR problem need to be modified so as to yield a better performance. 
    
    First, we introduce some additional notations for the algorithms that will be described below. We use $\U_f$ to denote the nodes on the path of flow $f$ that are included in $\U$, i.e., $\U_f = \V_{f} \cap \U$. Let $c^{\prime}_v$ denote the remaining capacity of VNF-node $v$, and let $c^{\prime}_{\U_i}$ denote the total remaining capacity of the set of VNF-nodes in $\U_i$, i.e., $c^{\prime}_{\U_i} = \sum_{v \in \U_i} c^{\prime}_{v}$. In what follows, we will introduce the MCA algorithm and the GCA algorithm.
    
    \subsubsection{Maximum-flow-based Capacity Allocation (MCA)}
    \label{subsubsec:MCA}
    
    \begin{algorithm}[!t]
        \caption{ The MCA algorithm}
        \label{alg:MCA}
        \begin{algorithmic}[1]
            \Statex Input: set of VNF-nodes $\U$, set of flows $\F_\U$, flow rates, and VNF-node capacities.
            \Statex Output: Capacity allocation $\boldsymbol{\lambda}$.
            \Statex \textbf{Phase I}:
            \State Obtain a basic optimal solution $\boldsymbol{\lambda}_{\U}$;
            \State $y_f^v \triangleq  \lambda_{f}^v/\lambda_{f}$, for all $\lambda_{f}^v$ in  $\boldsymbol{\lambda}_{\U}$;
            \State Assign each flow $f$ with $y_f^v=1$ to VNF-node $v$;
            \State Construct $G^\prime$ for the unassigned flows with positive $y_f^v$;
            \While {$G^\prime$ is not empty}
            \While {there is a singleton VNF-node in $G^\prime$}
            \State Perform the rounding in Step 1;
            \EndWhile
            \State Perform the rounding in Step 2;
            \EndWhile
            \Statex \textbf{Phase II}:
            \For {each flow $f$ in $\F_\U$ that is not assigned yet}
            \If {$c^\prime_{\U_f} \geq \lambda_{f}$ }
            \State Assign flow $f$ to a subset of VNF-nodes in $\U_f$;
            \State Update the remaining capacity of VNF-nodes $\U_f$;
            \EndIf
            \EndFor
            
        \end{algorithmic}
    \end{algorithm}
    
    We first present the MCA algorithm (Algorithm \ref{alg:MCA}), a capacity allocation algorithm based on the network flow formulation. The MCA algorithm has two phases. In Phase I, MCA makes allocation decisions by rounding a fractional flow assignment obtained by solving problem \eqref{eq:relaxedAllocation}; in Phase II, the remaining VNF-node capacities are allocated in a greedy manner. 
    
    Phase I: Let $\boldsymbol{\lambda}_{\U}$ be a flow assignment obtained from an optimal basic solution\footnote{A basic feasible solution is a solution that cannot be expressed as a convex combination of two feasible    solutions.} of problem  \eqref{eq:relaxedAllocation}, which can be obtained by solving a maximum flow problem as discussed earlier. We use $ y_f^v \triangleq  \lambda_{f}^v/\lambda_{f}$ to denote the fraction of flow $f$ assigned to VNF-node $v$ in the obtained solution $\boldsymbol{\lambda}_{\U}$. The algorithm begins with a temporary assignment of every flow $f$ with $y_f^v =1$ to the corresponding VNF-node $v$. For the remaining flows, we do the following. Let $G^\prime = (\F^\prime, \V^\prime, \E^\prime)$ be a bipartite graph constructed as follows. For each $\lambda_{f}^v \in \boldsymbol{\lambda}_{\U}$, if $0 < y_f^v < 1$, we add a flow vertex $f$ to the set $\F^\prime$, a VNF-node vertex $v$ to the set $\V^\prime$, and an edge, with weight $y_f^v$, connecting flow vertex $f$ to VNF-node vertex $v$, to the set $\E^\prime$. Note that graph $G^\prime$ cannot have a cycle because $\boldsymbol{\lambda}_{\U}$ is a basic feasible solution \cite[Lemma 5]{dawande2000approximation}. After constructing graph $G^\prime$, we repeatedly apply the following two steps to graph $G^\prime$ until it becomes empty. As a result, the modified flow assignment $y_f^v$ will become either zero or one.
    
    Step 1: For each VNF-node $v \in \V^\prime$ that has only one incident flow $f$ (called a singleton VNF-node), we modify its capacity allocation as follows. Let $r_v$ denote the total amount of flow rates assigned to VNF-node $v$ and let $r^\prime_v$ be the portion of $r_v$ contributed by fully assigned flows. Note that $r_v = r_v^\prime + \lambda_f^v$. If $r^\prime_{v} \geq \lambda_{f}^{v}$, then  we set $y_{f}^{v}$ to zero. Now, VNF-node $v$ has no incident edges to it, so we remove it from $G^\prime$. In this case, the value of solution $\boldsymbol{\lambda}_{\U}$ will be reduced by $\lambda_f^v$, which is no greater than $\frac{1}{2}r_v$. If $r^\prime_{v} < \lambda_{f}^{v}$, then we unassign the flows temporarily assigned to VNF-node $v$ and assign flow $f$ to VNF-node $v$ instead, i.e., set $y_f^v$ to one, and cancel the other fractions of flow $f$ assigned to other VNF-nodes. This is feasible because the rate of any flow is assumed to be no larger than the minimum VNF-node capacity. Then, we remove VNF-node $v$, flow $f$, and the associated edges from $G^\prime$. In this case, the value of solution $\boldsymbol{\lambda}_{\U}$ will be reduced by at most $r_v^\prime$, which is no greater than $\frac{1}{2}r_v$. We repeat Step 1 until no singleton VNF-node exists. Then, we go to Step 2.
    
    Step 2: In this step, we will perturb the fractional values of some edges in $G^\prime$ to make one of them either zero or one. The perturbation is designed such that the capacity and assignment constraints are not violated and the total assigned traffic remains the same. We describe the perturbation procedure in the following. Consider a VNF-node $v_1 \in \V^\prime$ that has a degree of at least two. Let $(v_1, f_1)$ and $(v_1, f_{k+1})$ denote two of the incident edges to VNF-node $v_1$. Let $\boldsymbol{p_1}$ and $\boldsymbol{p_2}$ denote the longest paths starting from VNF-node $v_1$ through edges $(v_1, f_1)$ and $(v_1, f_{k+1})$, respectively; such paths exist because $G^\prime$ is a forest. Here, we use $y_i^j$ to denote the fractional value of flow $i$ assigned to VNF-node $j$ and use $\lambda_j$ to denote the rate of flow $j$. Let $\boldsymbol{y_1} = (y_1^1, y_1^2, \dots, y_k^k)$ denote the fractional flow assignment on the edges of path $\boldsymbol{p_1}$, and let $f_1, \dots, f_k$ be the flow nodes of path $\boldsymbol{p_1}$. Similarly, let $\boldsymbol{y_2} = (y_{k+1}^{1}, y_{k+1}^{k+1}, \dots, y_{k+l}^{k+l-1})$ denote the fractional flow assignment on the edges of path $\boldsymbol{p_2}$, and let $f_{k+1}, \dots, f_{k+l}$ denote the flow nodes of path $\boldsymbol{p_2}$. We perturb $\boldsymbol{y_1}$ by adding to it $\boldsymbol{y^\prime_1} = (\frac{\lambda_{k}}{\lambda_{1}} \epsilon, - \frac{\lambda_{k}}{\lambda_{1}} \epsilon, \frac{\lambda_{k}}{\lambda_{2}} \epsilon, \dots, - \frac{\lambda_{k}}{\lambda_{k-1}} \epsilon, \epsilon)$, and we perturb $\boldsymbol{y_2}$ by adding to it $\boldsymbol{y^\prime_2} = (-\frac{\lambda_k}{\lambda_{k+1}} \epsilon, \frac{\lambda_k}{\lambda_{k+1}} \epsilon, - \frac{\lambda_k}{\lambda_{k+2}} \epsilon, \dots,  \frac{\lambda_k}{\lambda_{k+l-1}} \epsilon, -\frac{\lambda_k}{\lambda_{k+l}} \epsilon)$. We increase $\epsilon$ until one fractional value $y_f^v$ becomes zero or one, and if one, i.e., $y_f^v=1$, then we assign flow $f$ to the corresponding VNF-node $v$. An example to illustrate this step is shown in Fig. \ref{fig:perturbation}. In this new solution, at least one edge is removed from $G^\prime$. We repeat the perturbation procedure until at least one VNF-node becomes a singleton, and then we go back to Step 1. If $G^\prime$ becomes empty, we start Phase II.
    
    Phase II: We leverage the property that the traffic of a flow can be split and processed at multiple VNF-nodes. That is, after Phase I, we pick an unassigned flow $f$ and check if the total remaining capacity of VNF-nodes $\U_{f}$ is no smaller than $\lambda_{f}$. If so, we split flow $f$ so that the remaining capacities of some VNF-nodes in $\U_{f}$ can be used to fully process flow $f$ and assign flow $f$ to a subset of these VNF-nodes. Then, we update the remaining capacities of VNF-nodes $\U_{f}$. We repeat this  procedure until no more flow can be assigned.

\def\x{-2}
\def\y{-1}
\def\step{1.5}
\begin{figure}
    \centering
    \LARGE
    \resizebox{0.99\linewidth}{!}{
        
        \begin{tikzpicture}[transform shape]
        
        \begin{scope}[auto, every node/.style={minimum size=.5em}]
        \node[vertex](v) at (\x, \y) {};
        \node at (\x+0.6, \y) {$v_1$};
        \node at (\x+1.8, \y) {$ + $};
        \node[vertex](v1) at (\x, \y+\step) {};
        \node at (\x+0.6, \y+\step) {$ v_2 $};
        \node at (\x, \y+1.7*\step) {$ \vdots $};
        \node [minimum size=1.8em] (t1)at (\x-1.8, \y+1.8*\step) {};
        \node[vertex](vi) at (\x, \y+2.2*\step) {};
        \node at (\x+0.6, \y+2.2*\step) {$v_k$};
        \node[vertex](vi1) at (\x, \y-1*\step) {};
        \node at (\x+0.8, \y-1*\step) {$v_{k+1}$};
        \node at (\x, \y-1.4*\step) {$ \vdots $};
        \node[minimum size=1.8em] (t2)at (\x-1.8, \y-1.7*\step) {};
        \node[vertex](vij) at (\x, \y-2*\step) {};
        \node at (\x+0.8, \y-2*\step) {$v_{k+l}$};
        \def\x{-1.5}
        \node[vertex](f1) at (\x-4, \y+0.5*\step) {};
        \node at (\x-4.6, \y+0.5*\step) {$f_1$};
        \node[vertex](f2) at (\x-4, \y+1.5*\step) {};
        \node at (\x-4.6, \y+1.5*\step) {$ f_2 $};
        \node at (\x-4, \y+2.1*\step) {$ \vdots $};
        \node at (\x-2.3, \y+1.9*\step) {$ \vdots $};

        \node[vertex](fk) at (\x-4, \y+2.5*\step) {};
        \node at (\x-4.6, \y+2.5*\step) {$f_k$};
        \node [red] at (\x-5.2, \y+2.8*\step) {$\boldsymbol{p_1}$};
        \node[vertex](fk1) at (\x-4, \y-0.7*\step) {};
        \node at (\x-4.8, \y-0.7*\step) {$f_{k+1}$};
        \node [red] at (\x-5.2, \y-.2*\step) {$\boldsymbol{p_2}$};
        \node[vertex](fk2) at (\x-4, \y-1.5*\step) {};
        \node at (\x-4.8, \y-1.5*\step) {$ f_{k+2} $};
        \node[]  at (\x-4, \y-1.9*\step) {$\vdots$};
        \node at (\x-2.3, \y-1.6*\step) {$ \vdots $};
        \node[vertex](fkl) at (\x-4, \y-2.5*\step) {};
        \node at (\x-4.8, \y-2.5*\step) {$  f_{k+l} $};
        
        \begin{scope}[every path/.style={-}, every node/.style={inner sep=1pt}]
        \path (v) edge [ anchor= south, near start] node {$y_1^1$} (f1);
        \path (f1) edge [ anchor= south, near start] node {$y_1^2$} (v1);
        \path (v1) edge [ anchor= south, near start] node {$y_2^2$} (f2);
        \path (f2) edge [ anchor= south] node {} (t1);
        \path (t1) edge [ anchor= south] node {} (vi);
        \path (vi) edge [ anchor= south, near start] node {$y_k^k$} (fk);
        
        \path (v) edge [ anchor= south, near end] node {$y_{k+1}^1$} (fk1);
        \path (fk1) edge [ anchor= south, near end] node {$y_{k+1}^{k+1}$} (vi1);
        \path (vi1) edge [ anchor= south, near end] node {$y_{k+2}^{k+1}$} (fk2);
        \path (fk2) edge [ anchor= south] node {} (t2);
        \path (t2) edge [ anchor= south] node {} (vij);
        \path (vij) edge [ anchor= south, near end] node {$y_{k+l}^{k+l}$} (fkl);
        \end{scope}

        \def\x{5}
        \node[vertex](v) at (\x, \y) {};
        \node at (\x+0.6, \y) {$v_1$};
        \node[vertex](v1) at (\x, \y+\step) {};
        \node at (\x+0.6, \y+\step) {$ v_2 $};
        \node at (\x, \y+1.7*\step) {$ \vdots $};
        \node [minimum size=1.8em](t1)at (\x-1.8, \y+1.8*\step) {};
        \node[vertex](vi) at (\x, \y+2.2*\step) {};
        \node at (\x+0.6, \y+2.2*\step) {$v_k$};
        \node[vertex](vi1) at (\x, \y-1*\step) {};
        \node at (\x+0.8, \y-1*\step) {$v_{k+1}$};
        \node at (\x, \y-1.4*\step) {$ \vdots $};
        \node[minimum size=1.8em] (t2)at (\x-1.8, \y-1.7*\step) {};
        \node[vertex](vij) at (\x, \y-2*\step) {};
        \node at (\x+0.8, \y-2*\step) {$v_{k+l}$};
        \def\x{5.5}
        \node[vertex](f1) at (\x-4, \y+0.5*\step) {};
        \node at (\x-4.6, \y+0.5*\step) {$f_1$};
        \node[vertex](f2) at (\x-4, \y+1.5*\step) {};
        \node at (\x-4.6, \y+1.5*\step) {$ f_2 $};
        \node at (\x-4, \y+2.1*\step) {$ \vdots $};
        \node at (\x-2.3, \y+1.9*\step) {$ \vdots $};        \node[vertex](fk) at (\x-4, \y+2.5*\step) {};
        \node at (\x-4.6, \y+2.5*\step) {$f_k$};
        \node[vertex](fk1) at (\x-4, \y-0.7*\step) {};
        \node at (\x-4.8, \y-0.7*\step) {$f_{k+1}$};
        \node[vertex](fk2) at (\x-4, \y-1.5*\step) {};
        \node at (\x-4.8, \y-1.5*\step) {$ f_{k+2} $};
        \node[]  at (\x-4, \y-1.9*\step) {$\vdots$};
        \node at (\x-2.3, \y-1.6*\step) {$ \vdots $};        \node[vertex](fkl) at (\x-4, \y-2.5*\step) {};
        \node at (\x-4.8, \y-2.5*\step) {$  f_{k+l} $};

        \begin{scope}[every path/.style={-}, every node/.style={inner sep=1pt}]
        \path (v) edge [ anchor= south, near start] node {$\frac{\lambda_{k}}{\lambda_{1}} \epsilon$} (f1);
        \path (f1) edge [ anchor= south, near start] node {$- \frac{\lambda_{k}}{\lambda_{1}} \epsilon$} (v1);
        \path (v1) edge [ anchor= south, near start] node {$\frac{\lambda_{k}}{\lambda_{2}} \epsilon$} (f2);
        \path (f2) edge [ anchor= south] node {} (t1);
        \path (t1) edge [ anchor= south] node {} (vi);        
        \path (vi) edge [ anchor= south, near start] node {$\epsilon$} (fk);
        
        \path (v) edge [ anchor= south, near end] node {$-\frac{\lambda_k}{\lambda_{k+1}} \epsilon$} (fk1);
        \path (fk1) edge [ anchor= south, near end] node {$\frac{\lambda_k}{\lambda_{k+1}} \epsilon$} (vi1);
        \path (vi1) edge [ anchor= south, near end] node {$- \frac{\lambda_k}{\lambda_{k+2}} \epsilon$} (fk2);
        \path (fk2) edge [ anchor= south] node {} (t2);
        \path (t2) edge [ anchor= south] node {} (vij);
        \path (vij) edge [ anchor= south, near end] node {$-\frac{\lambda_k}{\lambda_{k+l}} \epsilon$} (fkl);
        \end{scope}
        \end{scope}
        
        \draw[red,thin,dashed] (\x-2,\y+2.2) ellipse  (3.7cm and 2.2cm);
        \draw[red,thin,dashed] (\x-2,\y-2.2) ellipse  (3.7cm and 2.2cm);
        \end{tikzpicture}}
    
    \caption{ An example of the edges perturbation}
    \label{fig:perturbation}
\end{figure}
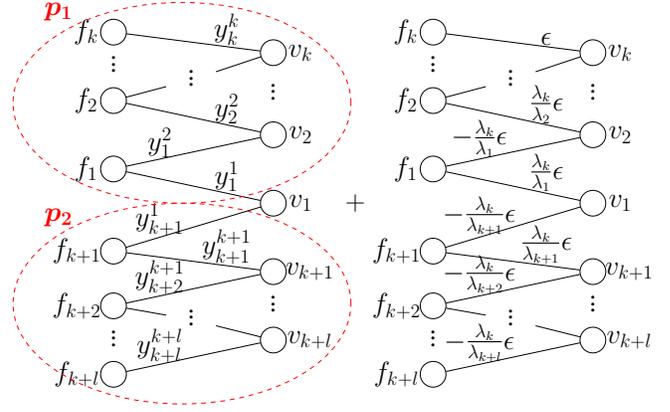

    We use $OPT(\text{\ref{eq:relaxedAllocation}}, \U)$ to denote the total traffic assigned to  a given set of VNF-nodes $\U$ by an optimal solution to problem \eqref{eq:relaxedAllocation}. Also, we use $\pi_{\text{MCA}}^\U$ to denote the total traffic assigned to VNF-nodes $\U$ by the MCA algorithm. The approximation ratio of the MCA algorithm is stated in the following Lemma.    \begin{lemma}
        \label{lemma:MCAapproximation}
        The MCA algorithm has an approximation ratio of $1/2$, i.e.,  $\pi_\text{MCA}^\U \geq \frac{1}{2} OPT(\text{\ref{eq:relaxedAllocation}}, \U)$. 
    \end{lemma}
    \begin{proof}
        See Appendix~\ref{proof:MCAapproximation1}.
    \end{proof}
    \subsubsection{Greedy Capacity Allocation (GCA)}
    \label{subsubsec:GCA} 
    While the MCA algorithm achieves an approximation ratio of $1/2$, it has a relatively high complexity of $O(F^2V^2)$ (refer to Table~\ref{table:complexity} for the complexity analysis). This high complexity may render the MCA algorithm unsuitable for certain scenarios in practice. To that end, we propose the GCA algorithm, a simple greedy capacity allocation algorithm that has a much lower complexity of $O(FV)$. A lower complexity  of the GCA algorithm is achieved at the cost of a slightly worse approximation ratio of $1/3$ (Lemma \ref{lemma:GCAapproximation1}). However, the approximation ratio of the GCA algorithm can be improved to $2/5$ (cf. Lemma \ref{lemma:GCAapproximation2} below) if an additional mild assumption (Assumption \ref{assump:equal_capacity}) holds.  The GCA algorithm has two phases. In Phase I, we sort flows of $\F_\U$ in a nonincreasing order  of their flow rates. Then, we iteratively go through the sorted list and assign each flow to any VNF-node in $\U_f$ if it has a sufficient capacity. In Phase II, the remaining capacities of the VNF-nodes can be allocated in a similar way to Phase II of the MCA algorithm  by leveraging the property that a flow can be processed at multiple VNF-nodes. However, here the remaining flows need to be considered according to the order in the sorted list $\F_\U$. The GCA algorithm is presented in Algorithm \ref{alg:GCA}.
 
     \begin{algorithm}[t]
        \caption{The GCA algorithm}
        \label{alg:GCA}
        \begin{algorithmic}[1]
            \Statex Input: set of VNF-nodes $\U$, set of flows $\F_\U$, flow rates, and VNF-node capacities.
            \Statex Output: Capacity allocation $\boldsymbol{\lambda}$.
            \State Sort flows $\F_\U$ in a noincreasing order of their flow rates;
            \Statex \textbf{Phase I}:
            \For {each flow $f$ in the sorted set $\F_\U$}
            \If {there is a VNF-node $v$ in $\U_f$ such that $c^\prime_v \geq \lambda_{f}$}
            \State Set $\lambda_{f}^v=\lambda_{f}$;
            \State Set $c^\prime_v = c^\prime_v -\lambda_{f}$;
            \EndIf
            \EndFor
            \Statex \textbf{Phase II}:
            \For {each flow $f$ in the sorted set $\F_\U$ that is not assigned yet}
            \If {$c^\prime_{\U_f} \geq \lambda_{f}$ }
            \State Assign flow $f$ to a subset of VNF-nodes in $\U_f$;
            \State Update the remaining capacity of VNF-nodes $\U_f$;
            \EndIf
            \EndFor
        \end{algorithmic}
    \end{algorithm}
    
    In Lemma \ref{lemma:GCAapproximation1}, we state the result about the approximation ratio of the GCA algorithm. We use $\pi_{\text{GCA}}^\U$ to denote the total traffic assigned to VNF-nodes $\U$ by the GCA algorithm.
    \begin{lemma}
        \label{lemma:GCAapproximation1}
        The GCA algorithm has an approximation ratio of 1/3, i.e., $\pi_\text{GCA}^\U \geq \frac{1}{3} OPT(\text{\ref{eq:relaxedAllocation}}, \U)$.
    \end{lemma}
    \begin{proof}
        See Appendix~\ref{proof:GCAapproximation1}.
    \end{proof}
    Further, we show in Lemma \ref{lemma:GCAapproximation2} that the approximation ratio of the GCA algorithm can be improved to $2/5$ when an additional mild assumption (Assumption \ref{assump:equal_capacity}) holds.
    \begin{assumption}
        \label{assump:equal_capacity}
        Assume that all the VNF-nodes in $\U$ have the same capacity and that every flow $f$ in $\F_\U$  traverses at least two VNF-nodes in $\U$, i.e., $|\V_f \cap \U| \geq 2$. 
    \end{assumption}
    \begin{lemma}
        \label{lemma:GCAapproximation2}
        Suppose that Assumption \ref{assump:equal_capacity} holds. Then, the GCA algorithm has an improved approximation ratio of $2/5$, i.e.,  $\pi_\text{GCA}^\U \geq \frac{2}{5} OPT(\text{\ref{eq:relaxedAllocation}}, \U)$. 
    \end{lemma}
    \begin{proof}
        See Appendix~\ref{proof:GCAapproximation2}.
    \end{proof}
    \subsection{Main Results}
    \label{subsec:mainResult}
    We state our main results in Theorems \ref{theorem:mainResult} and \ref{theorem:mainResult2}.
    \begin{theorem}
        \label{theorem:mainResult}
        The RP-MCA algorithm has an approximation ratio of $\frac{1}{2}(1-1/e)$ for problem \eqref{eq:mainProblem}. 
    \end{theorem}
    \begin{proof}
         The RP-MCA algorithm has two main components: 1) VNF-nodes placement and 2) capacity allocation. For the relaxed placement subproblem \eqref{eq:relaxedPlacement}, we use $\pi_{G}^\U$ to denote the value of the optimal relaxed allocation for the set of VNF-nodes $\U$ selected by the SG algorithm or the EG algorithm. Also, we use $OPT(P)$ to denote the optimal value of any problem $(P)$. We have the following result:
        \begin{equation}  
        \label{eq:placementresult2}      
        \begin{aligned}
        \pi_{G}^\U&  \stackrel{\text{(a)}}{\geq}  (1-1/e)OPT(\text{\ref{eq:relaxedPlacement}})\\
        & \stackrel{\text{(b)}}{=} (1-1/e)OPT(\text{\ref{eq:relaxedProblem}}) \\
        & \stackrel{\text{(c)}}{\geq} (1-1/e)OPT(\text{\ref{eq:mainProblem}}),
        \end{aligned}
        \end{equation}
        where (a) is due to Lemma \ref{lemma:placement}, (b) holds because an optimal capacity allocation is assumed for the objective function of problem \eqref{eq:relaxedPlacement}, and (c) holds because problem \eqref{eq:relaxedProblem} is a relaxed version of problem \eqref{eq:mainProblem}. 
        
        The second component of the RP-MCA algorithm is the capacity allocation using the MCA algorithm for the set of VNF-nodes $\U$ selected by the SG or EG algorithm. We have the following result:
        \begin{equation}
        \label{eq:mainTheorem}
        \begin{aligned}
        \pi_\text{MCA}^\U & \stackrel{\text{(a)}}{\geq} \frac{1}{2}OPT(\text{\ref{eq:relaxedAllocation}}, \U) \\
        & \stackrel{\text{(b)}}{=} \frac{1}{2}\pi_{G}^\U \\
        & \stackrel{\text{(c)}}{\geq} \frac{1}{2} (1-1/e)OPT(\text{\ref{eq:mainProblem}}),
        \end{aligned}
        \end{equation} 
        where (a) comes from the approximation ratio of the MCA algorithm in Lemma \ref{lemma:MCAapproximation}, (b) holds because when $\pi^{\U}_{G}$ is obtained for problem \eqref{eq:relaxedPlacement} using the greedy algorithms, an optimal capacity allocation (i.e., an optimal solution to problem \eqref{eq:relaxedAllocation} associated with the considered $\mathcal{U}$) is assumed for the objective function, and (c) holds from Eq. \eqref{eq:placementresult2}. 
        Therefore, the result of Theorem \ref{theorem:mainResult} follows.
        \end{proof}
        
    \begin{theorem}
        \label{theorem:mainResult2}
        The RP-GCA algorithm has an approximation ratio of $\frac{1}{3}(1-1/e)$ for problem \eqref{eq:mainProblem}. 
    \end{theorem}
        \begin{proof}
        The proof follows the same argument as in the proof of Theorem \ref{theorem:mainResult}. Since the GCA algorithm achieves an approximation ratio of $1/3$ for the capacity allocation (Lemma \ref{lemma:GCAapproximation1}), the proof proceeds exactly the same except that we need to replace $1/2$ with $1/3$ in Eq. \eqref{eq:mainTheorem}.
    \end{proof}

            \begin{table}[t]
        \centering
        \begin{tabular}{| c | c |p{1.6cm} | c |}
            \hline
            Setting & Algorithm & Approximation & Complexity \\
            \hline
            \multirow{3}{5.4em}{Homogeneous \\ VNF costs} & RP-MCA & $\frac{1}{2}(1-1/e)$ & $O(kV)^{\dagger}\, + \, O(F^2V^2)$ \\ [1ex]  \cline{2-4}  
            
            & RP-GCA & $\frac{1}{3}(1-1/e)$ & $O(kV)^\dagger\, + \, O(FV)$  \\ [1ex] 
            
            & & $\frac{2}{5}(1-1/e)^*$ & \\ [1ex]  \hline
            \multirow{3}{5.4em}{Heterogeneous \\ VNF costs}     & RP-MCA & $\frac{1}{2}(1-1/e)$ & $O(V^5)^{\dagger}\, + \, O(F^2V^2)$ \\ [1ex]  \cline{2-4} 
            
            & RP-GCA & $\frac{1}{3}(1-1/e)$ & $O(V^5)^{\dagger}\, + \, O(FV)$ \\ [1ex]  
            &  & $\frac{2}{5}(1-1/e)^*$ &  \\ [1ex]  \hline
            
        \end{tabular}
        \caption{Approximation ratios and time complexities of the proposed algorithms.  \footnotesize{$^*$These are the approximation results for the GCA algorithm when Assumption \ref{assump:equal_capacity} holds.} \footnotesize{$^\dagger$This is the number of function evaluations used in the submodular optimization.}}
        \label{table:complexity}
    \end{table}
 
    Table \ref{table:complexity} summarizes the complexity of our proposed algorithms. In literature on submodular optimization, the complexity of algorithms for submodular functions is often measured through the number of function evaluations. The function evaluation itself is usually assumed to be conducted by an oracle,  and thus its complexity is not taken into account \cite{li2018towards}. We followed this approach here. Note that we can utilize other alternative algorithms to the EG algorithm to improve the running time  substantially but with a slightly worse approximation ratio \cite{khuller1999budgeted, li2018towards}. 
    We provide more discussions about the complexity analysis in Appendix \ref{appendix:complexity}.
    %We provide more discussions about the complexity analysis along with the tradeoff between the performance and complexity in Appendix \ref{appendix:complexity}.

    \section{Numerical Results}
             \begin{figure}[t]
        \centering
        \begin{subfigure}{0.48\linewidth}
            \centering
            \includegraphics[width=0.99\linewidth]{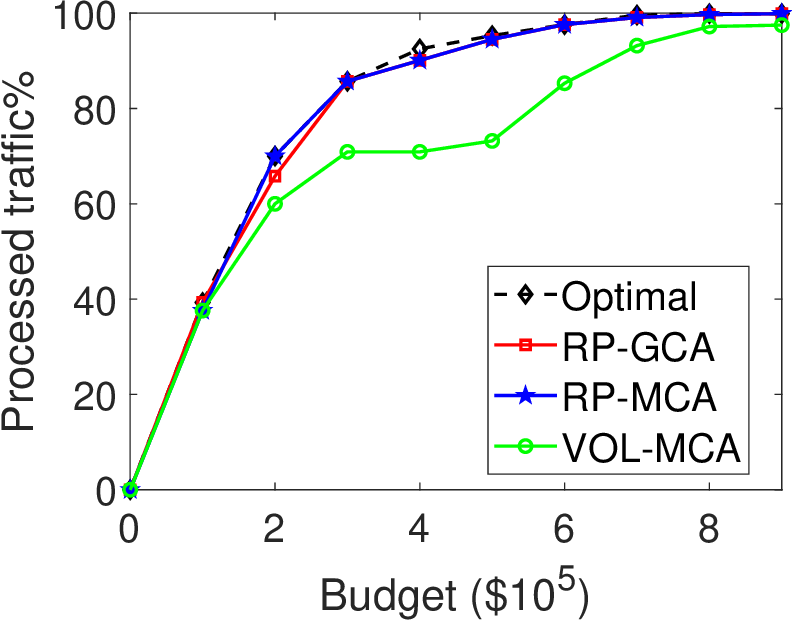}
            \caption{VNF-node capacity = 1 Gbps}
            \label{fig:Abilene_TotalFlow_diff_budget}
        \end{subfigure}
        \begin{subfigure}{0.48\linewidth}
            \centering
            \includegraphics[width=0.99\linewidth]{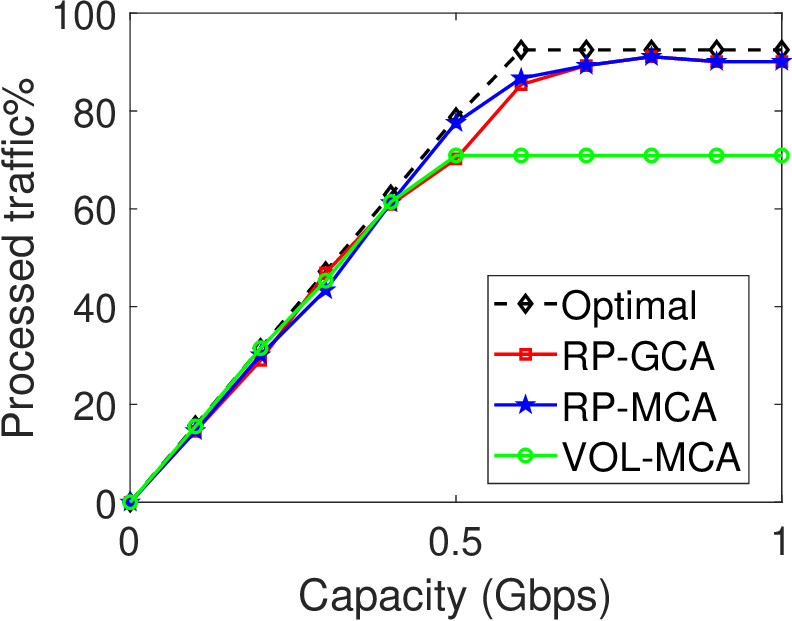}
            \caption{Budget = $\$400K$}
            \label{fig:Abilene_TotalFlow_diff_capacities}
        \end{subfigure}
         
        \caption{Evaluation on the Abilene dataset}
        %(a) Different budget with VNF-node capacity = 1 Gbps. (b) Different capacities with budget = $\$400K$.}
        \label{fig:abilene}
    \end{figure}

    \label{sec:evaluation}
    In order to evaluate the performance of the proposed algorithms, we consider real-world network topologies and traffic statistics. We compare the proposed algorithms with the following baselines: 1) optimal solution: we can solve problem \eqref{eq:mainProblem} optimally using  Gurobi~\cite{gurobi}, an ILP solver, for the presented instances. 
    Recall that the VPCA problem is NP-hard in general (Theorem~\ref{theorem:nphardness}).
    Although we are able to obtain the optimal solution for the problem instances we consider here, it may take a prohibitively large amount of time to obtain the optimal solution for some other problem instances. 2) VOL-MCA \cite{hong2016incremental}: this scheme selects the nodes with the highest traffic volume that traverses them. For the selected nodes, we allocate their capacity using the proposed MCA algorithm. We evaluate the performance of each algorithm based on the percentage of the processed traffic, which is defined as the ratio between the total volume of the traffic fully processed by the VNF-nodes and the total traffic volume. We run the simulations on a PC with Intel Core i7-7700 processor and 32GB memory.
    
    \subsection{Evaluation Datasets and Simulation Parameters}
  
    \subsubsection{Abilene Dataset}
    We consider the Abilene dataset \cite{abilene} collected from an educational backbone network in North America. The network consists of 12 nodes and 144 flows. Each flow rate was recorded every five minutes for 6 months. The OSPF weights were recorded, which allows us to compute the shortest path of each flow based on these weights. In our experiments, we set the flow rate to the recorded value of the first day at 8:00 pm. 
    
    \subsubsection{SNDlib Datasets}
    We also consider two datasets from SNDlib \cite{orlowski2010sndlib}: Cost266 and Ta2. Cost266 has 37 nodes and 1332 flows; Ta2 has 65 nodes and 1869 flows. For Cost266, the routing cost of each link is available, which can be used to compute the shortest path of each flow. For Ta2, we use the hop-count-based shortest path. We set the capacity of each VNF-node to $1$ Gbps as the default value. 
    
     \subsection{Evaluation Results}
      \begin{figure*}[ht]
     \begin{minipage}[l]{.8\columnwidth}
         \centering
         \includegraphics[width=0.80\linewidth]{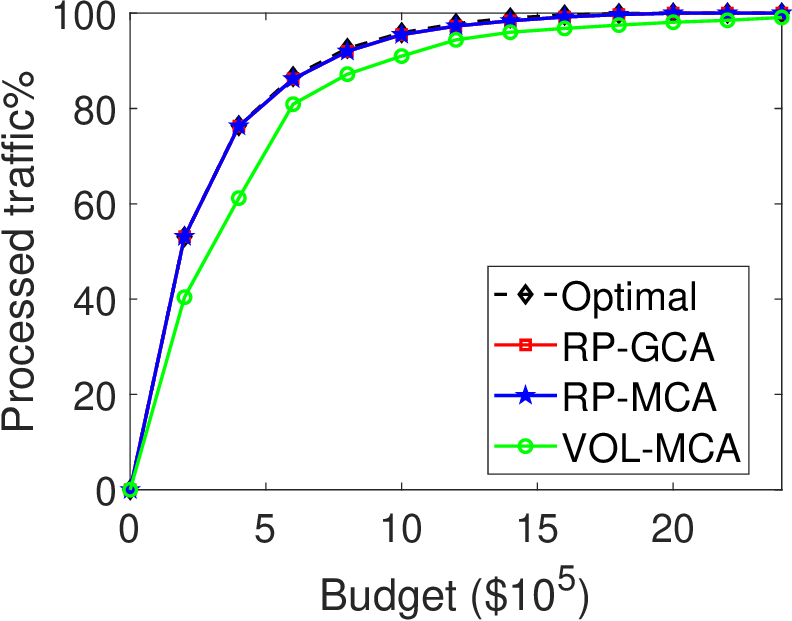}
         \caption{Evaluation on the Cost266 datasets}\label{fig:cost266}
     \end{minipage}
     \hfill{}
     \begin{minipage}[r]{1.2\columnwidth}
          \centering
        \begin{subfigure}{0.48\linewidth}
            \centering
            \includegraphics[width=0.99\linewidth]{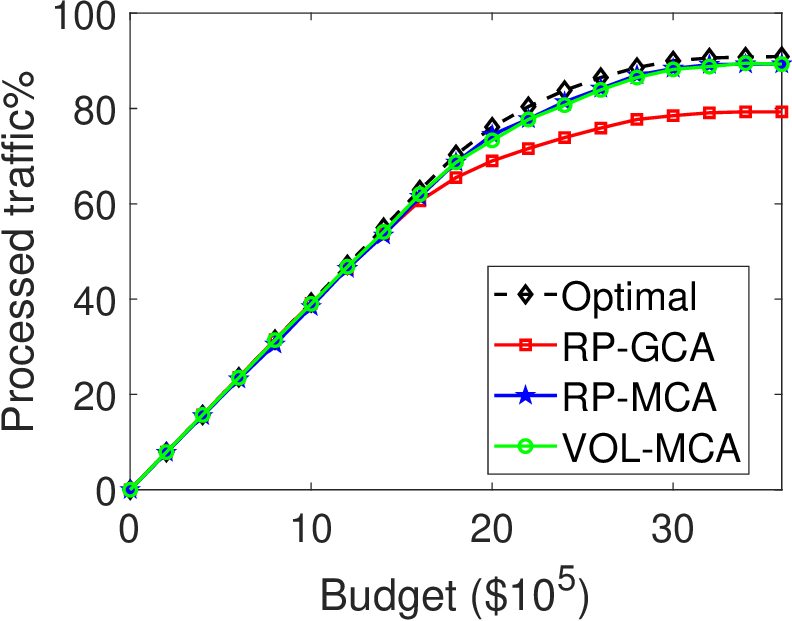}
            \caption{Capacity = 1 Gbps}
            \label{fig:ta2_TotalFlow_diff_budget_1gbps}
        \end{subfigure}
        \begin{subfigure}{0.48\linewidth}
            \centering
            \includegraphics[width=0.99\linewidth]{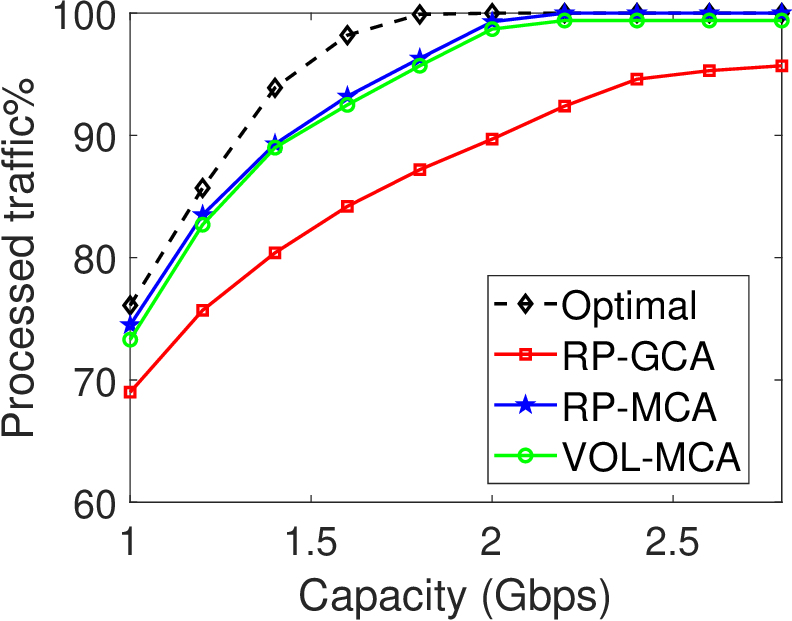}
            \caption{Budget = $\$2M$}
            \label{fig:ta2_TotalFlow_diff_capacity_budget_16_starts_with_1gbps}
        \end{subfigure}
        \caption{Evaluation on the Ta2 datasets}
        \label{fig:ta2}
     \end{minipage}
\end{figure*}

   We start with the Abilene dataset, where we set the cost of a VNF-node to $\$100K$ and vary the processing capacity between $100$ Mbps and $1$ Gbps. Also, we vary the total budget between $\$100K$ and $\$900K$, with an increment of $\$100K$. Fig.~\ref{fig:abilene}(\subref{fig:Abilene_TotalFlow_diff_budget}) shows the percentage of processed traffic for the considered algorithms. We can see that both the RP-MCA and RP-GCA algorithms perform almost the same as the optimal solution and have up to 20$\%$ improvement over the VOL-MCA algorithm. Note that as the budget increases, the total processed traffic increases under all the considered algorithms. However, while the proposed algorithms need a budget of  $\$500K$ to process around $95\%$ of the flows, the VOL-MCA algorithm actually requires around $\$800K$ to process the same amount. While the total amount of processed traffic can be improved by deploying more VNF-nodes, another option is to consider provisioning more capacities at each VNF-node. In other words, we can deploy fewer but more powerful VNF-nodes, which can improve the percentage of the total processed traffic as well. We can make this observation in Fig.~\ref{fig:abilene}(\subref{fig:Abilene_TotalFlow_diff_capacities}), where we vary the capacity of each VNF-node between $100$ Mbps and $1$ Gbps. The total processed traffic reaches around $92\%$ and saturates at that point. This happens because the path of some flows may not contain any of the VNF-nodes placed by the algorithms we consider.

    Furthermore, we evaluate the proposed algorithms on datasets with a larger number of nodes and flows. We also set the cost of each VNF-node to $\$100K$ and the processing capacity to $1$ Gbps. We start with the Cost266 dataset, which consists of 37 nodes and 1332 flows. Since we have a large number of flows, we select flows uniformly at random. Specifically, we select 1000 flows and repeat each experiment 10 times to gain more statistical significance.  We also vary the budget between $\$100K$ and $\$2.4M$ with an increment of $\$200K$. Fig. \ref{fig:cost266} shows that the proposed algorithms still exhibit superior performance compared to the VOL-MCA algorithm and match the performance of the optimal solution. An interesting observation is that for this specific instance, we only need to deploy six VNF-nodes to process around $90\%$ of the traffic. This suggests that we may gradually transition to the NFV paradigm with a lower total cost and consolidate the NFV functionalities at a small portion of nodes. The later point of having a small number of VNF-nodes is advantageous in terms of reducing the management burden, but that could increase the risk of having a single point of failure. In our future work, we will investigate how to adjust the proposed algorithms to ensure resilience against hardware failures at the VNF-nodes. Note that VNF-node capacities are not a bottleneck because flow rates of the Cost266 dataset are relatively small. Therefore, increasing the capacity of VNF-nodes shows no impact on the total processed traffic.
    
    Finally, we consider a denser topology Ta2, which consists of 65 nodes and 1869 flows. Similarly, we select 1500 flows uniformly at random and repeat each experiment 10 times. We present the results in Fig. \ref{fig:ta2}(\subref{fig:ta2_TotalFlow_diff_budget_1gbps}). The simulation results show that both RP-MCA and VOL-MCA have a similar performance, which is also close to the optimal. An interesting observation is that all the considered algorithms tend to select the same set of nodes, which results in a similar performance of the RP-MCA and VOL-MCA as both of them use the same capacity allocation algorithm. In addition, the GCA algorithm performs worse than the MCA algorithm for this particular instance, which results in a worst overall performance under the RP-GCA algorithm. The maximum total processed traffic reaches around $92\%$ and saturates at that point. The reason is that while all the nodes on the paths of the flows have been selected, they do not have sufficient processing capacity to process all the flows. In this case, selecting more other nodes (due to an increased budget) does not help.
    This motivates us to study the impact of increasing the VNF-node capacity in Fig.~\ref{fig:ta2}(\subref{fig:ta2_TotalFlow_diff_capacity_budget_16_starts_with_1gbps}). We consider a fixed budget of $\$2M$ and vary the capacity of each node between $1$ Gbps and $2.8$ Gbps with an increment of $200$ Mbps. We can see that as the capacity increases, the total processed traffic increases and can reach around  $100\%$ when the capacity is more than $2$ Gbps.

    \section{Conclusion and Future Work}
    \label{sec:conclusion}
    In this paper, we studied the problem of deploying VNF-nodes and allocating their capacities. We showed how to overcome the non-submodularity of the problem by introducing a novel relaxation method. By utilizing a decomposition of the problem and a novel network flow reformulation, we were able to prove submodularity of the relaxed placement subproblem and develop efficient algorithms with constant approximation ratios for the original problem. Through extensive evaluations based on trace-driven simulations, we showed that the proposed algorithms have a performance close to the optimal solution and better than a heuristic algorithm. 
    
    Our work also raises several interesting questions that are worth investigating as future work. 
    First, we would like to consider the problem of joint placement, routing, and capacity allocation and investigate the impact of dynamic routing on the objective function. Despite the expected additional challenges due to routing, especially with integral resources and service function chaining as shown in \cite{Feng2017}, we plan to further investigate whether our proposed framework can be extended to more general settings with joint placement and routing. 
    In this regard, recent work in~\cite{poularakis2020approximation} provides multi-criteria approximation algorithms for minimum cost joint placement and integral routing of service function chains under storage, computation, and communication capacity constraints. Understanding if our framework can be used to provide efficient approximations in similar settings with budget constraints is of interest for future work.
    % As we mentioned in the Related Work Section, we have shown that our proposed framework can be extended to joint placement and resource allocation problems in settings with multiple network functions and multiple types of resources \cite{sallam2019placement}. 
    %
    Second, while we have assumed that each VNF-node has a fixed capacity in this paper,
    it is worth studying the optimization of the VNF-node capacity.
    A straightforward extension is to consider the setting where each VNF-node hosts multiple servers, each of which has a different cost and a different processing capacity. Our proposed framework can be modified to accommodate the extension: we consider multiple replications of each node, called virtual nodes, each of which corresponds to a server with its own cost and processing capacity. The path of each flow now needs to be updated to include the corresponding virtual nodes as well. We further elaborate on the proposed extension and present some additional simulation results in our online technical report~\cite{sallam2019joint}.

    \bibliographystyle{IEEEtran}
    \bibliography{references}

\begin{thebibliography}{10}
\providecommand{\url}[1]{#1}
\csname url@samestyle\endcsname
\providecommand{\newblock}{\relax}
\providecommand{\bibinfo}[2]{#2}
\providecommand{\BIBentrySTDinterwordspacing}{\spaceskip=0pt\relax}
\providecommand{\BIBentryALTinterwordstretchfactor}{4}
\providecommand{\BIBentryALTinterwordspacing}{\spaceskip=\fontdimen2\font plus
\BIBentryALTinterwordstretchfactor\fontdimen3\font minus
  \fontdimen4\font\relax}
\providecommand{\BIBforeignlanguage}[2]{{%
\expandafter\ifx\csname l@#1\endcsname\relax
\typeout{** WARNING: IEEEtran.bst: No hyphenation pattern has been}%
\typeout{** loaded for the language `#1'. Using the pattern for}%
\typeout{** the default language instead.}%
\else
\language=\csname l@#1\endcsname
\fi
#2}}
\providecommand{\BIBdecl}{\relax}
\BIBdecl

\bibitem{chiosi2012network}
M.~Chiosi, D.~Clarke, P.~Willis, A.~Reid, J.~Feger, M.~Bugenhagen, W.~Khan,
  M.~Fargano, C.~Cui, H.~Deng \emph{et~al.}, ``Network functions
  virtualisation: An introduction, benefits, enablers, challenges and call for
  action,'' in \emph{SDN and OpenFlow World Congress}, vol.~48.\hskip 1em plus
  0.5em minus 0.4em\relax sn, 2012.

\bibitem{poularakis2017one}
K.~Poularakis, G.~Iosifidis, G.~Smaragdakis, and L.~Tassiulas, ``One step at a
  time: Optimizing sdn upgrades in isp networks,'' in \emph{Proceedings of IEEE
  INFOCOM}, 2017.

\bibitem{Amdocs_whitepaper}
Amdocs, ``{Bringing NFV to Life - Technological and Operational Challenges in
  Implementing NFV},'' \emph{White paper}, 2016.

\bibitem{sang2017provably}
Y.~Sang, B.~Ji, G.~R. Gupta, X.~Du, and L.~Ye, ``Provably efficient algorithms
  for joint placement and allocation of virtual network functions,'' in
  \emph{Proceedings of IEEE INFOCOM}, 2017.

\bibitem{lukovszki2018approximate}
T.~Lukovszki, M.~Rost, and S.~Schmid, ``Approximate and incremental network
  function placement,'' \emph{Journal of Parallel and Distributed Computing},
  2018.

\bibitem{hong2016incremental}
D.~K. Hong, Y.~Ma, S.~Banerjee, and Z.~M. Mao, ``Incremental deployment of sdn
  in hybrid enterprise and isp networks,'' in \emph{Proceedings of the
  Symposium on SDN Research}.\hskip 1em plus 0.5em minus 0.4em\relax ACM, 2016.

\bibitem{shi2018competitive}
M.~Shi, X.~Lin, S.~Fahmy, and D.-H. Shin, ``Competitive online convex
  optimization with switching costs and ramp constraints,'' in
  \emph{Proceedings of IEEE INFOCOM}, 2018.

\bibitem{feng2018optimal}
H.~Feng, J.~Llorca, A.~M. Tulino, and A.~F. Molisch, ``Optimal dynamic cloud
  network control,'' \emph{IEEE/ACM Transactions on Networking}, vol.~26,
  no.~5, pp. 2118--2131, 2018.

\bibitem{Feng2017}
H.~Feng, J.~Llorca, A.~M. Tulino, D.~Raz, and A.~F. Molisch, ``Approximation
  algorithms for the nfv service distribution problem,'' in \emph{Proceedings
  of IEEE INFOCOM}, 2017.

\bibitem{chen2018virtual}
Y.~Chen, J.~Wu, and B.~Ji, ``Virtual network function deployment in
  tree-structured networks,'' in \emph{IEEE 26th International Conference on
  Network Protocols (ICNP)}, 2018.

\bibitem{tomassilli2018provably}
A.~Tomassilli, F.~Giroire, N.~Huin, and S.~P{\'e}rennes, ``Provably efficient
  algorithms for placement of service function chains with ordering
  constraints,'' in \emph{Proceedings of IEEE INFOCOM}, 2018.

\bibitem{lukovszki2015online}
T.~Lukovszki and S.~Schmid, ``Online admission control and embedding of service
  chains,'' in \emph{International Colloquium on Structural Information and
  Communication Complexity}.\hskip 1em plus 0.5em minus 0.4em\relax Springer,
  2015.

\bibitem{sallam2018shortest}
G.~Sallam, G.~R. Gupta, B.~Li, and B.~Ji, ``Shortest path and maximum flow
  problems under service function chaining constraints,'' in \emph{Proceedings
  of IEEE INFOCOM}, 2018.

\bibitem{He2018It}
T.~He, H.~Khamfroush, S.~Wang, T.~La~Porta, and S.~Stein, ``It's hard to share:
  Joint service placement and request scheduling in edge clouds with sharable
  and non-sharable resources,'' in \emph{IEEE ICDCS}, 2018.

\bibitem{kablan2017stateless}
M.~Kablan, A.~Alsudais, E.~Keller, and F.~Le, ``Stateless network functions:
  Breaking the tight coupling of state and processing,'' in \emph{14th {USENIX}
  Symposium on Networked Systems Design and Implementation ({NSDI} 17)}, 2017,
  pp. 97--112.

\bibitem{poularakis2020service}
K.~Poularakis, J.~Llorca, A.~M. Tulino, I.~Taylor, and L.~Tassiulas, ``Service
  placement and request routing in mec networks with storage, computation, and
  communication constraints,'' \emph{IEEE/ACM Transactions on Networking},
  2020.

\bibitem{rost2019virtual}
M.~Rost and S.~Schmid, ``Virtual network embedding approximations: Leveraging
  randomized rounding,'' \emph{IEEE/ACM Transactions on Networking}, vol.~27,
  no.~5, pp. 2071--2084, 2019.

\bibitem{nemeth2020cost}
B.~Nemeth, Y.-A. Pignolet, M.~Rost, S.~Schmid, and B.~Vass, ``Cost-efficient
  embedding of virtual networks with and without routing flexibility,''
  \emph{IFIP Networking Conference (Networking)}, 2020.

\bibitem{nemhauser1981maximizing}
G.~L. Nemhauser and L.~A. Wolsey, ``Maximizing submodular set functions:
  formulations and analysis of algorithms,'' in \emph{North-Holland Mathematics
  Studies}.\hskip 1em plus 0.5em minus 0.4em\relax Elsevier, 1981, vol.~59, pp.
  279--301.

\bibitem{khuller1999budgeted}
S.~Khuller, A.~Moss, and J.~S. Naor, ``The budgeted maximum coverage problem,''
  \emph{Information Processing Letters}, vol.~70, no.~1, pp. 39--45, 1999.

\bibitem{das2011submodular}
A.~Das and D.~Kempe, ``Submodular meets spectral: Greedy algorithms for subset
  selection, sparse approximation and dictionary selection,'' \emph{arXiv
  preprint arXiv:1102.3975}, 2011.

\bibitem{chen2017weakly}
L.~Chen, M.~Feldman, and A.~Karbasi, ``Weakly submodular maximization beyond
  cardinality constraints: Does randomization help greedy?'' \emph{arXiv
  preprint arXiv:1707.04347}, 2017.

\bibitem{feldman2014constrained}
M.~Feldman and R.~Izsak, ``Constrained monotone function maximization and the
  supermodular degree,'' \emph{arXiv preprint arXiv:1407.6328}, 2014.

\bibitem{sallam2019placement}
G.~Sallam, Z.~Zheng, and B.~Ji, ``Placement and allocation of virtual network
  functions: Multi-dimensional case,'' in \emph{2019 IEEE 27th International
  Conference on Network Protocols (ICNP)}.\hskip 1em plus 0.5em minus
  0.4em\relax IEEE, 2019, pp. 1--11.

\bibitem{goldberg2014efficient}
A.~V. Goldberg and R.~E. Tarjan, ``Efficient maximum flow algorithms,''
  \emph{Communications of the ACM}, vol.~57, no.~8, pp. 82--89, 2014.

\bibitem{bach2013learning}
F.~Bach, ``Learning with submodular functions: A convex optimization
  perspective,'' \emph{Foundations and Trends{\textregistered} in Machine
  Learning}, vol.~6, no. 2-3, pp. 145--373, 2013.

\bibitem{dawande2000approximation}
M.~Dawande, J.~Kalagnanam, P.~Keskinocak, F.~S. Salman, and R.~Ravi,
  ``Approximation algorithms for the multiple knapsack problem with assignment
  restrictions,'' \emph{Journal of combinatorial optimization}, vol.~4, no.~2,
  pp. 171--186, 2000.

\bibitem{li2018towards}
W.~Li and N.~Shroff, ``Towards practical constrained monotone submodular
  maximization,'' \emph{arXiv preprint arXiv:1804.08178}, 2018.

\bibitem{gurobi}
``Gurobi, https://www.gurobi.com/.''

\bibitem{abilene}
``Abilene dataset, http://www.cs.utexas.edu/~yzhang/research/abilenetm/.''

\bibitem{orlowski2010sndlib}
S.~Orlowski, R.~Wess{\"a}ly, M.~Pi{\'o}ro, and A.~Tomaszewski, ``Sndlib
  1.0-survivable network design library,'' \emph{Networks: An International
  Journal}, vol.~55, no.~3, pp. 276--286, 2010.

\bibitem{poularakis2020approximation}
K.~Poularakis, J.~Llorca, A.~M. Tulino, and L.~Tassiulas, ``Approximation
  algorithms for data-intensive service chain embedding,'' in \emph{Proceedings
  of MobiHoc 2020}, 2020, pp. 131--140.

\bibitem{sallam2019joint}
G.~Sallam and B.~Ji, ``Joint placement and allocation of vnf nodes with budget
  and capacity constraints,'' \emph{arXiv preprint
  https://arxiv.org/abs/1901.03931}, 2019.

\end{thebibliography}
    
    \begin{appendices}
        \section{Proof of Theorem \ref{theorem:nphardness}}
        \label{proof:nphardness}
        \begin{proof}
            We start by proving that the allocation subproblem~\eqref{eq:allocation} is NP-hard.
            The proof is by a reduction from a special case of the single knapsack (SK) problem, where for each item the profit and the weight are identical. In the SK problem, we have a knapsack $k$ and a set of items $\mathcal{I}$. The knapsack has a capacity $W$, and each item $i \in \mathcal{I}$ has a weight of $w_i$, which is the same as the profit. The objective is to find a subset of items $\mathcal{I^\prime} \subseteq \mathcal{I}$ that has the maximum total profit and can be packed in the knapsack without exceeding its capacity. Given an arbitrary instance $\mathcal{A}=(k, \mathcal{I})$ of the SK problem, we construct an instance $\mathcal{D}=(\V, \F)$ of the allocation problem \eqref{eq:allocation}. The set $\V$ has only one node $v_1$ with a capacity that is equal to the capacity of the knapsack $k$. Each flow $f \in \mathcal{F}$ corresponds to an item in $i \in \mathcal{I}$. A flow $f$ has a traffic rate $\lambda_f$ that is equal to the corresponding item weight $w_i$. Node $v_1$ is the only VNF-node, and all the flows traverse node $v_1$. If we can solve the instance $\mathcal{D}$  of problem \eqref{eq:allocation}, the subset of flows assigned to node $v_1$, which has the maximum total traffic rate, can be mapped to the corresponding items and solve the instance $\mathcal{A}$ of SK problem. Similarly, a solution for instance $\mathcal{A}$ of the SK problem can be mapped to a solution for instance $\mathcal{D}$ by simply mapping the selected items $\mathcal{I^\prime}$ to the corresponding flows that solve the instance $\mathcal{D}$ of problem \eqref{eq:allocation}.
        %\end{proof}
        
        Next, we prove the NP-hardness of the placement subproblem \eqref{eq:placement}.
        %\begin{proof}
            The proof is by a reduction from the budgeted maximum coverage (BMC) problem. In the BMC problem, we have a set of points $\mathcal{M}$ and a set of candidate locations $\mathcal{S}$. Each point $m \in \mathcal{M}$ has a weight of $w_m$. Each location $s \in \mathcal{S}$ has a cost of $b_s$ and covers a subset of points $\mathcal{M}_s \subseteq \mathcal{M}$. The objective is to select a subset of locations $\mathcal{S}^{\prime} \subseteq \mathcal{S}$ such that the total weight of the points covered by at least one location in $\mathcal{S}^{\prime}$ is maximized while the total cost of the selected locations does not exceed a given budget $B$. Given an arbitrary instance $\mathcal{A}=(\mathcal{M}, \mathcal{S}, B)$ of BMC, we will construct an instance $\mathcal{D}=(\F, \V, B)$ of problem \eqref{eq:placement} as follows. Each flow $f \in \mathcal{F}$ corresponds to a point $m \in \mathcal{M}$; the rate of a flow is equal to the weight of the corresponding point. Each node $v \in \mathcal{V}$ corresponds to a location $s \in \mathcal{S}$; the cost of a node is equal to that of the corresponding location. The path of a flow consists of the nodes corresponding to the locations that cover the point corresponding to this flow. The deployment budget of the instance $\mathcal{D}$ is equal to the budget of the instance $\mathcal{A}$. All the nodes have an infinite capacity.  We will show that a solution for the instance $\mathcal{D}$ exists if and only if a solution for the instance $\mathcal{A}$ exists. If we can solve the instance $\mathcal{A}$ of BMC, the subset of locations $\mathcal{S^\prime} \subseteq \mathcal{S}$ that solves $\mathcal{A}$ of BMC can be mapped to the corresponding nodes in $\V$ to become VNF-nodes and solves the instance $\mathcal{D}$ of problem \eqref{eq:allocation}. Similarly, if we solve the instance $\mathcal{D}$, then the obtained set of VNF-nodes can be mapped to the corresponding subset of locations that solve the instance $\mathcal{A}$ of BMC.
        \end{proof}
       
        \section{Proof of Lemma \ref{lemma:Q2MaxFlow}}
        \label{proof:Q2MaxFlow}
        \begin{proof}
        Recall that ${R_3}(\U)=\max_{\boldsymbol{\lambda} \in \Lambda^{\U}} {R_2^\U}(\boldsymbol{\lambda})$, where $\Lambda^{\U}$ is the set of assignments satisfying the capacity constraint \eqref{eq:nodecapacity} and the flow rate constraint \eqref{eq:traffic2}.
        It suffices to show the following: 
        \begin{enumerate}[(A)]
        \item for any assignment $\boldsymbol{\lambda} \in \Lambda^{\U}$, one can construct an $s$-$\V$ flow $\varphi \in \overline{\F}$ such that ${R_2^\U}(\boldsymbol{\lambda}) = \Phi(\N, \U) - \Phi(\U, \N)$; 
        
        \item for any $s$-$\V$ flow $\varphi \in \overline{\F}$, one can construct an assignment $\boldsymbol{\lambda} \in \Lambda^{\U}$ such that ${R_2^\U}(\boldsymbol{\lambda}) = \Phi(\N, \U) - \Phi(\U, \N)$.
        \end{enumerate}
        Note that Part (A) implies 
        %\[
        $\max_{\boldsymbol{\lambda} \in \Lambda^{\U}} {R_2^\U}(\boldsymbol{\lambda}) \leq \max_{\varphi \in \overline{\F}} (\Phi(\N, \U) - \Phi(\U, \N))$
        %\]
        and Part (B) implies 
        %\[
        $\max_{\boldsymbol{\lambda} \in \Lambda^{\U}} {R_2^\U}(\boldsymbol{\lambda}) \geq \max_{\varphi \in \overline{\F}} (\Phi(\N, \U) - \Phi(\U, \N))$,
        %\]
        which lead to Eq.~\eqref{eq:equivalence}.
    
        We first show Part (A). For any assignment $\boldsymbol{\lambda} \in \Lambda^{\U}$, we construct a function $\varphi \in \overline{\F}$ in the following manner: 
        \begin{enumerate}[(i)] 
        \item set $\varphi(s, f) = \sum_{v \in \V_f} \lambda_f^v$ for each edge $(s, f) \in \link_1$;
        \item set $\varphi(f, v^\prime) = \lambda_f^v$ for each  edge $(f, v^\prime) \in \link_2$;
        \item set $\varphi(v^\prime, v) = \sum_{f \in \F} \lambda_f^v$ for each edge$(v^\prime, v) \in \link_3$.
        \end{enumerate}
        Note that $\lambda_f^v =0$ for all $v \notin \U$. It is easy to verify that constraints \eqref{eq:nodecapacity} and \eqref{eq:traffic2} imply that the constructed function $\varphi$ is an $s$-$\V$ flow. Further, the following is also satisfied:
        \begin{equation}  
        \label{eq:R2ToNetFlow}      
        \begin{aligned}
        {R_2^\U}(\boldsymbol{\lambda}) & \stackrel{\text{(a)}}{=} \sum_{f \in \F} \sum_{v \in \V_f \cap \U} \lambda_f^v \\
        & \stackrel{\text{(b)}}{=} \sum_{f \in \F} \sum_{v \in \V_f} \lambda_f^v \\
        & \stackrel{\text{(c)}}{=} \sum_{f \in \F} \varphi(s, f) \\        
        & \stackrel{\text{(d)}}{=} \sum_{f \in \N_\F} \varphi(s, f) \\
        & \stackrel{\text{(e)}}{=}  \Phi(\{s\}, \N) - \Phi(\N, \{s\}) \\  
        & \stackrel{\text{(f)}}{=} \Phi(\N, \N_\V) - \Phi(\N_\V, \N) \\
        & \stackrel{\text{(g)}}{=} \Phi(\N, \U) - \Phi(\U, \N),
        \end{aligned}
        \end{equation}
        where (a) is from the definition of ${R_2^\U}(\boldsymbol{\lambda})$, (b) is from $\lambda_f^v = 0$ for all $v \notin \U$, (c) is from (i), (d) is from the one-to-one mapping between $\F$ and $\N_\F$, (e) is from the definition of $\Phi(\{s\}, \N)$ and $\Phi(\N, \{s\})=0$, (f) holds because the net-flow at the source $s$ plus the net-flow at the sinks is equal to zero, and (g) holds because no flow goes to the sinks in $\N_\V \setminus \U$.
        
        We now show Part (B). For any $s$-$\V$ flow $\varphi \in \overline{\F}$, we first obtain another $s$-$\V$ flow $\varphi^{\prime} \in \overline{\F}$ by deleting all the flow going to the sinks in $\N_\V \setminus \U$. Note that this procedure does not change the net-flow at the sinks in $\U$, i.e., $\Phi^{\prime}(\N, \U) - \Phi^{\prime}(\U, \N) = \Phi(\N, \U) - \Phi(\U, \N)$, where $\Phi^{\prime}$ corresponds to $\varphi^{\prime}$. Then, we construct an assignment $\boldsymbol{\lambda} \in \Lambda^{\U}$ by simply setting $\lambda_f^v = \varphi^{\prime}(f, v^\prime)$ for every $f \in \F$ and every $v \in \U$.  
        It is easy to verify that the definition of the $s$-$\V$ flow implies that constraints \eqref{eq:nodecapacity} and \eqref{eq:traffic2} are satisfied for assignment $\boldsymbol{\lambda}$. Finally, following the same steps in Eq.~\eqref{eq:R2ToNetFlow}, we can show ${R_2^\U}(\boldsymbol{\lambda}) = \Phi^{\prime}(\N, \U) - \Phi^{\prime}(\U, \N)$, and thus, ${R_2^\U}(\boldsymbol{\lambda}) = \Phi(\N, \U) - \Phi(\U, \N)$.
        
        Combining Parts (A) and (B) completes the proof. 
        \end{proof}

        \section{Proof of Lemma \ref{lemma:MCAapproximation}}
        \label{proof:MCAapproximation1}
        \begin{proof}
            In Phase I, the algorithm repeatedly alternates between Step 1 and Step 2. Step 1 will be repeated $m$ times, where $m \leq |V|$. In each repetition of Step 1, the capacity allocation of a given VNF-node $v$ is modified, and the value of $OPT(\text{\ref{eq:relaxedAllocation}}, \U)$ is reduced by at most $\frac{1}{2}r_v$. In Step 2, the perturbation does not change the value of the total assigned traffic. Therefore, the following is satisfied after Phase I:
            \begin{equation}
            \begin{aligned}
            \pi_{\text{MCA}}^\U & \geq OPT(\text{\ref{eq:relaxedAllocation}}, \U) - \frac{1}{2} \sum_{i=1}^{m} r_{v_i} \\
            & \geq  OPT(\text{\ref{eq:relaxedAllocation}}, \U) - \frac{1}{2}OPT(\text{\ref{eq:relaxedAllocation}}, \U) \\
            & = \frac{1}{2}OPT(\text{\ref{eq:relaxedAllocation}}, \U).
            \end{aligned}
            \end{equation}
            That is, the total traffic assigned to VNF-nodes $\U$ after Phase I is $\pi_{\text{MCA}}^\U \geq \frac{1}{2}OPT(\text{\ref{eq:relaxedAllocation}}, \U)$.    In Phase II, the total assigned traffic will either increase or remain the same in the worst case. Therefore, the result of the Lemma follows.
        \end{proof}
        \section{Proof of Lemma \ref{lemma:GCAapproximation1}}
        \label{proof:GCAapproximation1}
        \begin{proof}
            
            Let $\F^\prime \subseteq \F_{\U}$ denote the set of unassigned flows after the end of algorithm \ref{alg:GCA} and $\U^\prime = \cup_{f \in \F^\prime_\U} \U_f$ be the set of candidate VNF-nodes for the unassigned flows $\F^\prime$. We remind the reader that we use $c_u$ (resp. $c_{\U_i}$) to denote the capacity of VNF-node $u$ (resp. VNF-nodes $\U_i$). Similarly, we use $r_u$ (resp. $r_{\U_i}$) to denote the total traffic assigned to VNF-node $u$ (resp. VNF-nodes $\U_i$). In Lemma \ref{lemma:halfcapacity},  we start by showing  that the total assigned traffic to any VNF-node in $\U^\prime$ is at least half of its total capacity, i.e., $r_u \geq \frac{1}{2} c_u$ for any VNF-node $u$ in $\U^\prime$, and use that to prove the $1/3$ approximation ratio of Lemma~\ref{lemma:GCAapproximation1}. 
            \begin{lemma}
                \label{lemma:halfcapacity}
                After Phase I of the GCA algorithm, it holds that for any VNF-node $u$ in $\U^\prime$, $r_u \geq \frac{1}{2} c_u$.
            \end{lemma}
            
            \begin{proof}
                We will prove this by contradiction. Let say, for the sake of contradiction, that there is an unassigned flow $f$ for which there is a VNF-node $u$ in $\U_f$ such that $r_u < \frac{1}{2} c_u$. This means the flows assigned to VNF-node $u$ have traffic rate less than $\frac{1}{2}c_u$. Furthermore, since flow $f$ was not assigned to VNF-node $u$, its traffic rate has to be greater than half the capacity of node $u$, i.e., $\lambda_{f} > \frac{1}{2}c_u$.  However, that contradicts our algorithm where flows with the highest traffic rate are considered first, and thus flow $f$ would be assigned to VNF-node $u$ instead of some of the already assigned flows.
            \end{proof}
            
            Next, the proof of Lemma \ref{lemma:GCAapproximation1} proceeds as follows. The maximum traffic that can be assigned by any algorithm to VNF-nodes $\U$ has the following upper bound:
            \begin{equation}
            \begin{aligned}
            OPT(\text{\ref{eq:relaxedAllocation}}, \U) &\stackrel{\text{(a)}}{\leq}  \pi_\text{GCA}^\U + c_{\U^\prime}\\
            & \stackrel{\text{(b)}}{\leq} \pi_\text{GCA}^\U + 2r_{\U^\prime}\\
            & \stackrel{\text{(c)}}{\leq} \pi_\text{GCA}^\U + 2\pi_\text{GCA}^\U\\
            & = 3\pi_\text{GCA}^\U,
            \end{aligned}
            \end{equation} 
            where (a) holds because  the maximum traffic that can be assigned by an optimal solution is at most the sum of the traffic of the assigned flows, which is $\pi_\text{GCA}^\U$, and the maximum possible traffic that can be assigned for the unassigned flows, which is $c_{\U^\prime}$; (b) holds from Lemma \ref{lemma:halfcapacity} because the total traffic assigned to VNF-nodes in $\U^\prime$ by Algorithm \ref{alg:GCA} is at least half of their total capacity, i.e., $r_{\U^\prime} \geq \frac{1}{2} c_{\U^\prime}$; (c) holds because $r_{\U^\prime}$ is upper bounded by $\pi_\text{GCA}^\U$.
            This completes the proof.
        \end{proof}
        \section{Proof of Lemma \ref{lemma:GCAapproximation2}}
        \label{proof:GCAapproximation2}
        \begin{proof}
            We first present Lemma \ref{lemma:twothirdcapacity}, but before that we repeat Assumption \ref{assump:equal_capacity} here to ease the proof navigation.
            \addtocounter{assumption}{-1}
            \begin{assumption}
        Assume that all the VNF-nodes in $\U$ have the same capacity and that every flow $f$ in $\F_\U$  traverses at least two VNF-nodes in $\U$, i.e., $|\V_f \cap \U| \geq 2$. 
    \end{assumption}
            \begin{lemma}
                \label{lemma:twothirdcapacity}
                If Assumption \ref{assump:equal_capacity} holds, then for any unassigned flow $f \in \F^\prime$, it holds that for any pair of VNF-nodes $(u, v)$ in $\U_f$, $r_{\{u, v\}} \geq \frac{2}{3} c_{\{u, v\}}$.
            \end{lemma}
            \begin{proof}
                Since the capacity of all VNF-nodes is the same by Assumption \ref{assump:equal_capacity}, we will use the symbol $c$ to denote the capacity of any VNF-node. We prove this lemma by contradiction. Assume that there is an unassigned flow $f$ for which there is a pair of VNF-nodes $(u, v)$ in $\U_f$ such that $r_{\{u, v\}} < \frac{2}{3} c_{\{u, v\}}$. In this case, the rate of flow $f$ has to be greater than $\frac{2}{3}c$; otherwise, it would fit on the combined remaining capacities of VNF-nodes $u$ and $v$ and would be assigned in Phase II of the algorithm. However, this also means that the flows assigned to VNF-nodes $u$ and $v$ have a rate that is less than $\frac{2}{3}c$, which contradicts our algorithm where flows with larger traffic rate will be considered first, and if possible get assigned. 
            \end{proof}
            The rest of the proof of Lemma \ref{lemma:GCAapproximation2} follows the same argument as in the proof of Lemma \ref{lemma:GCAapproximation1} with the difference that we have $r_{\U^\prime} \geq \frac{2}{3} c_{\U^\prime}$ by Lemma \ref{lemma:twothirdcapacity}. The other parts of the proof are the same. That is, 
            \begin{equation}
            \begin{aligned}
            OPT(\text{\ref{eq:relaxedAllocation}}, \U) &\leq  \pi_\text{GCA}^\U + c_{\U^\prime}\\
            & \leq \pi_\text{GCA}^\U + \frac{3}{2}r_{\U^\prime}\\
            & \leq \pi_\text{GCA}^\U + \frac{3}{2}\pi_\text{GCA}^\U\\
            & \leq \frac{5}{2}\pi_\text{GCA}^\U.
            \end{aligned}
            \end{equation} 
            This completes the proof.
        \end{proof}
        
      \section{Complexity Analysis}
        \label{appendix:complexity}
        In this section, we analyze the complexity of the algorithms presented in Table \ref{table:complexity}.
        Each of the proposed algorithms has two sequential components: placement and capacity allocation. We analyze the complexity of each component in the following.
        
        \textbf{Complexity of the Placement Algorithms.} We have two placement algorithms: the SG algorithm and the EG algorithm. 
        For the SG algorithm, in each iteration, we select a new node, which requires $O(V)$ functions evaluations. Since we can select at most $k$ nodes (due to the budget limit), we need $O(kV)$ function evaluations in total. 
        %The SG algorithm requires $O(kV)$ function evaluations, while the EG algorithm requires $O(V^5)$ function evaluations.
        For the EG algorithm, we have two phases: in Phase I, we need $O(V^2)$ function evaluations to evaluate all subsets of size one or two; in Phase II, we need $O(V^5)$ function evaluations to evaluate all subsets of size three and augment each subset in a greedy manner. Hence, the overall complexity of the EG algorithm is $O(V^5)$.
        
        \textbf{Complexity of the Capacity Allocation Algorithms.} We have two capacity allocation algorithms: the MCA algorithm and the GCA algorithm. For the MCA algorithm, we have two phases. In Phase I, we start by solving a maximum flow problem for a graph with $O(F+V)$ vertices, which has a complexity of $O(F^3)$ if solved using the Push-relabel algorithm \cite{goldberg2014efficient}. Then, the algorithm proceeds by repeatedly alternating between implementing Step 1 and Step 2. Step 1 is executed for at most $V$ times; each execution of Step 1 has a complexity of $O(F)$. Therefore, the overall complexity of Step 1 is $O(FV)$. For Step 2, we remove at least one edge at a time using the edge perturbation. Since the total number of edges is at most $O(FV)$, each execution of Step 2 has a complexity of $O(FV)$. Hence, the complexity of repeating Step 2 for at most $O(FV)$ edges is $O(F^2V^2)$. The overall complexity of Phase I is $O(F^2V^2)$ as Step 2 dominates. In Phase II, for each of the (at most $F$) unassigned flows, the algorithm tries to assign it to a subset of nodes in $\V$, so Phase II has a complexity of $O(FV)$. The overall complexity of the MCA algorithm is $O(F^2V^2)$ as Phase I dominates. For the GCA algorithm, in Phase I, the sorting operation has a complexity of $O(F\log F)$, and assigning each flow to only one VNF-node has a complexity of $O(FV)$. Phase II is the same as that of the MCA algorithm, which has a complexity of $O(FV)$. Therefore, the complexity of the GCA algorithm is $O(FV)$. 

    \end{appendices}

 \begin{IEEEbiography}[{\includegraphics[width=1in,height=1.25in,clip,keepaspectratio]{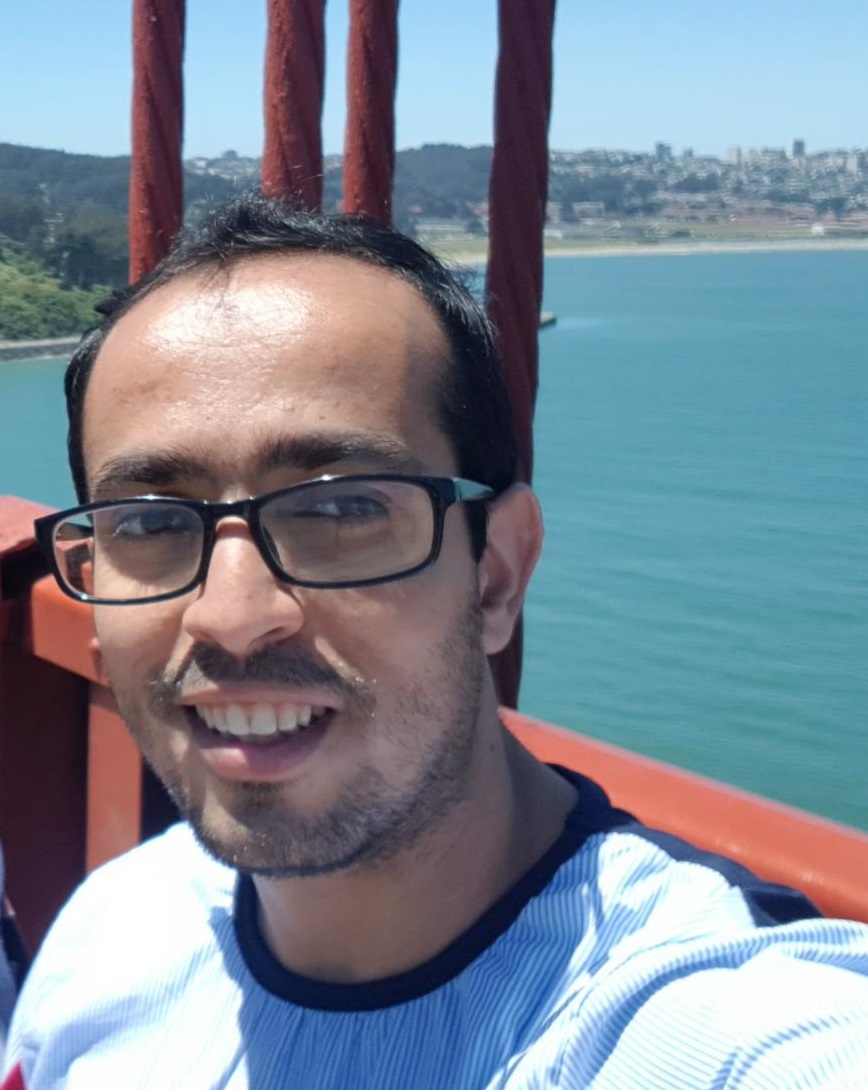}}]{Gamal Sallam}
 obtained the B.Sc. degree in Information Technology and Systems from Cairo University in 2011, the M.Sc. degree in Computer Networks from King Fahd University of Petroleum and Minerals in 2016, and the PhD degree in Computer Science from Temple University in 2020. His current research interests include resource allocation in network function virtualization. He received the Outstanding Research Assistant Award from the Department of Computer and Information Sciences and from the College of Science and Technology of Temple University in 2018 and 2019, respectively. 
 \end{IEEEbiography}

% \begin{IEEEbiography}[{\includegraphics[width=1in,height=1.25in,clip,keepaspectratio]{BoJi_Bio.jpg}}]{Bo Ji}(S'11-M'12-SM'18)
% received his B.E. and M.E. degrees in Information Science and Electronic Engineering from Zhejiang University, Hangzhou, China, in 2004 and 2006, respectively, and his Ph.D. degree in Electrical and Computer Engineering from The Ohio State University, Columbus, OH, USA, in 2012. Dr. Ji is an Associate Professor in Department of Computer Science at Virginia Tech, Blacksburg, VA, USA. Prior to joining Virginia Tech, he was an Associate Professor in Department of Computer and Information Sciences and a faculty member of the Center for Networked Computing at Temple University, where he was an Assistant Professor from July 2014 to June 2020. He was also a Senior Member of the Technical Staff with AT\&T Labs, San Ramon, CA, from January 2013 to June 2014. His research interests are in the modeling, analysis, control, optimization, and learning of computer and network systems, such as wired and wireless networks, large-scale IoT systems, high performance computing systems and data centers, 
% and cyber-physical systems. Dr. Ji is a senior member of the IEEE and a member of the ACM. He is a National Science Foundation (NSF) CAREER awardee (2017) and an NSF CISE Research Initiation Initiative (CRII) awardee (2017). He is also a recipient of the IEEE INFOCOM 2019 Best Paper Award.
% \end{IEEEbiography}        

\begin{IEEEbiography}[{\includegraphics[width=1in,height=1.25in,clip,keepaspectratio]{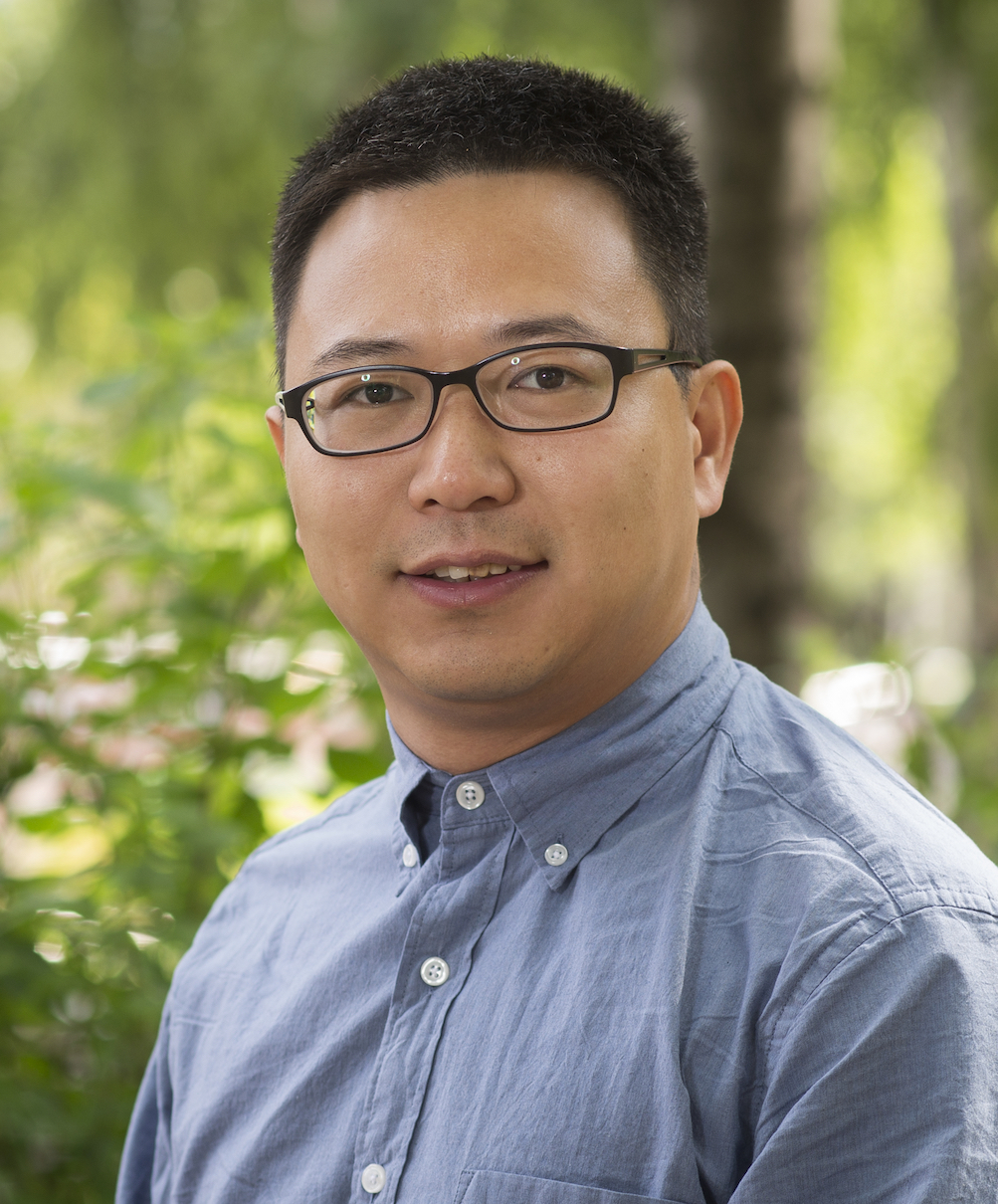}}]{Bo Ji}(S'11-M'12-SM'18)
received his B.E. and M.E. degrees in Information Science and Electronic Engineering from Zhejiang University, Hangzhou, China, in 2004 and 2006, respectively, and his Ph.D. degree in Electrical and Computer Engineering from The Ohio State University, Columbus, OH, USA, in 2012. Dr. Ji is an Associate Professor in the Department of Computer Science at Virginia Tech, Blacksburg, VA, USA. Prior to joining Virginia Tech, he was an Associate/Assistant Professor in the Department of Computer and Information Sciences at Temple University from July 2014 to July 2020. He was also a Senior Member of the Technical Staff with AT\&T Labs, San Ramon, CA, from January 2013 to June 2014. His research interests are in the modeling, analysis, control, and optimization of computer and network systems, such as wired and wireless networks, large-scale IoT systems, high performance computing systems and data centers, and cyber-physical systems. He currently serves on the editorial boards of the IEEE/ACM Transactions on Networking, IEEE Transactions on Network Science and Engineering, IEEE Internet of Things Journal, and IEEE Open Journal of the Communications Society. Dr. Ji is a senior member of the IEEE and a member of the ACM. He is a National Science Foundation (NSF) CAREER awardee (2017) and an NSF CISE Research Initiation Initiative (CRII) awardee (2017). He is also a recipient of the IEEE INFOCOM 2019 Best Paper Award.
\end{IEEEbiography}

\end{document}